\documentclass[12pt]{article}
\usepackage{graphicx,amssymb,amsfonts,latexsym,amsmath,amsthm,times}
\usepackage{epsfig}
\usepackage{fancyhdr}
\usepackage{color}
\usepackage[english]{babel}
\usepackage{subfigure}
\numberwithin{equation}{section}

\newtheorem{theorem}{Theorem}

\newtheorem{proposition}[theorem]{Proposition}
\newtheorem{remark}[theorem]{Remark}

\def\XC{\mathcal{X}}

\def\M{\mathbf{M}}

\def\R{\mathbf{R}}
\def\S{\mathbf{S}}

\def\Z{\mathbf{Z}}

\def\c{\mathbf{c}}

\def\x{\mathbf{x}}

\def\y{\mathbf{y}}
\def\1{\mathbf{1}}
\def\z{\mathbf{z}}

\def\al{\alpha}
\def\be{\beta}
\def\pa{\partial}
\def\ep{\epsilon}
\def\de{\delta}

\newcommand{\la}{\lambda}

\newcommand{\om}{\omega}

\newcommand{\De}{\Delta}

\newcommand{\Om}{\Omega}

\begin{document}

\vspace*{2cm}\centerline{\Large \bf Generalized Mass Action Law and}

 \vspace{0.2cm}
\centerline{\Large \bf Thermodynamics of Nonlinear Markov Processes} \vspace*{1cm}

\centerline{\bf A.N. Gorban$^a$\footnote{Corresponding author. E-mail: ag153@le.ac.uk},
V.N. Kolokoltsov$^b$}

\vspace{0.5cm}


\vspace*{0.5cm}


\centerline{$^a$Department of Mathematics, University of Leicester, Leicester, LE1 7RH,
UK}

\centerline{$^b$Department of Statistics, University of Warwick,
 Coventry, CV4 7AL, UK}


\vspace*{1cm}

\noindent {\bf Abstract.}The  nonlinear Markov processes are the measure-valued dynamical
systems which preserve positivity. They can be represented as the law of large numbers
limits of general Markov models of interacting particles. In physics, the kinetic
equations allow Lyapunov functionals (entropy, free energy, etc.). This may be considered
as a sort of inheritance of the Lyapunov functionals from the microscopic master
equations. We study nonlinear Markov processes that inherit thermodynamic properties from
the microscopic linear Markov processes. We develop the thermodynamics of nonlinear
Markov processes and analyze the asymptotic assumption, which are sufficient for this
inheritance.

\noindent {\bf Key words:}{Markov process; nonlinear kinetics; Lyapunov functional;
entropy; quasiequilibrium; quasi steady state}


\noindent {\bf AMS subject classification:}{ 80A3; 60J25; 60J60; 60J75; 82B40}



\tableofcontents
\newpage


\section{Introduction}

\subsection{What is the proper nonlinear generalization of the Markov processes?}

First order kinetics (the Kolmogorov--Chapman or master equation) is used in defining of
nonlinear kinetic equations: the microscopic dynamics is replaced by Markov processes and
then the large linear system is reduced to a nonlinear kinetics of some moments with
referring to the law of large numbers for the stochastic evolution. This approach became
very popular after the works of Kac \cite{Kac1956} and Prigogine and Balescu
\cite{Prigogine1959}. In this sense, the Markov processes serve as a source of nonlinear
kinetics. The stochastic simulation of chemical reactions \cite{Gillespie1976} made the
master equation approach to kinetics more popular in many applications.

At the same time, master equation is considered as a simplest kinetic equation because it
typically defines a contraction semigroup. If we consider, for example, Markov
transitions between a finite number of states, $A_i \to A_j$, then the probability
distribution relaxes exponentially  to an equilibrium and if the digraph of transitions
is connected then the normalized equilibrium distribution is unique. On contrary, the
interaction between states may produce nonlinear kinetic equations with various
non-trivial dynamic effects. For example, if we write for two states (`rabbits' and
`foxes') $\mbox{rabbit}\to 2\, \mbox{rabbits}$, $\mbox{fox}+\mbox{rabbit}\to
(1+a)\,\mbox{foxes}$ (the interaction step), and $\mbox{fox}\to \emptyset$, and apply the
standard mass action law then we get the predator--pray Lotka--Volterra system with
oscillations.

The classical {\em  mass action law (MAL)} systems are dense among the differential
equations which preserve positivity (different versions of this theorem are proven in
\cite{Korzukhin1967,GorbanBykYabEssays1986}, see also discussion in \cite{Kowalski1993}).
Therefore, if we aim to consider a general class of kinetic equations which includes the
MAL systems then the only important restriction is preservation of positivity. On this
way we approach the theory of nonlinear Markov processes \cite{BeKo,Ko10}.

In general spaces of states, the  nonlinear Markov processes are the measure-valued
dynamical systems which preserve positivity. They can be represented as the law of large
numbers limits of general Markov models of interacting particles. The sensitivity
analysis for these nonlinear evolution equations, that is the systematic study of the
smooth dependence on the initial conditions and other parameters via the study of
linearized system around a solution were performed in \cite{Ko10,Ko06,Ko07}.

Linear Markov chains have many Lyapunov functionals. For a finite chain with equilibrium
distribution $P^*=(p_i^*)$ they have the form
\begin{equation}
H_h(P\|P^*)=\sum_i p_i^* h\left(\frac{p_i}{p_i^*}\right),
\end{equation}
where $P=(p_i)$ is the current distribution and $h$ is an arbitrary convex function on
the positive semi-axis. These functionals were discovered by R\'enyi in 1960
\cite{Renyi1961} and studied further by Csisz\'ar \cite{Csiszar1963}, Morimoto
\cite{Morimoto1963} and many other authors (see review in \cite{GorGorJudge2010}). The
functions $H_h(P(t)\|P^*)$ monotonically decrease (non-increasing) with time on the
solutions $P(t)$ of the corresponding master equations. Proposition \ref{MorimotoHtheoremext}
of Appendix extends the Morimoto result to continuous state models.

In physics, the kinetic equations allow Lyapunov functionals (entropy, free energy,
etc.). This may be considered as a sort of inheritance of the Lyapunov functionals from
the microscopic master equations. In this paper, we study nonlinear Markov processes that
inherit thermodynamic properties from the microscopic linear Markov processes. We develop
the thermodynamics of nonlinear Markov processes and analyze the asymptotic assumption,
which are sufficient for this inheritance.

\subsection{Preliminaries: MAL, detailed balance and $H$-theorems}
The classical thermodynamics follows the Clausius laws \cite{Clausius1865}
\begin{enumerate}
\item{The energy of the Universe is constant.}
\item{The entropy of the Universe tends to a maximum.}
\end{enumerate}
In practice, we assume that the `Universe' is the minimal system, which is isolated with acceptable precision and includes the system of interest.

Kinetics is expected be concordant with the laws of thermodynamics. In physical kinetics,
Boltzmann's $H$ theorem established a link between the statistical entropy of
one-particle distribution function in gas kinetics and the thermodynamic entropy
\cite{Boltzmann1872}. Boltzmann's proof of his $H$-theorem used the principle of detailed
balance: At equilibrium, each collision is equilibrated by the reverse collision. This
principle is based on the {\em microscopic reversibility}: the Newton equation of motion
for particles are invariant with respect to a time reversal and a the space inversion
transformations. Five years before Boltzmann, Maxwell considered detailed balance as a
consequence of the principle of sufficient reason \cite{Maxwell1867}. Later on, this
principle was declared as a new fundamental law \cite{Lewis1925}. For modern proofs and
refutations of detailed balance we refer to \cite{GorbanResPhys2014}.

After Boltzmann, new kinetic equations in physics are always to be tested for concordance
with the laws of thermodynamics. Many particular $H$-theorems have been proved for
various classes of kinetic equations. The principle of detailed balance has been widely
used in these proofs. For MAL with detailed balance, the $H$-function and the entropy
production formula are very similar to the Boltzmann equation with detailed balance. Let
$A_1, \ldots , A_n$ be the components. For any set of non-negative numbers $\alpha_{\rho
i}, \, \beta_{\rho i}\geq 0$ ($i=1,\ldots , n$, $\rho=1, \ldots , m$) a reversible
reaction mechanism is given by the system of formal equations:
\begin{equation}\label{reversible mechanism}
\alpha_{\rho 1}A_1+ \ldots + \alpha_{\rho n} A_n \rightleftharpoons \beta_{\rho 1}A_1+ \ldots + \beta_{\rho n} A_n \, .
\end{equation}
According to the principle of detailed balance, each reaction has an inverse one and we join them in one reversible reaction.
A non-negative real variable, concentration $c_i$, is associated with each component $A_i$, two positive constants, rate constants $k^{\pm}_{rho}$ are associated with  each elementary reaction and reaction rates are defined as
\begin{equation}\label{reaction rates}
r^+_{\rho}=k^+_{\rho}\prod_{i=1}^n c_i^{\alpha_{\rho i}},\;\;r^-_{\rho}=k^-_{\rho}\prod_{i=1}^n c_i^{\beta_{\rho i}}, \;\; r_{\rho}=r^+_{\rho}-r^-_{\rho} \, .
\end{equation}
The reaction kinetics MAL equations are
\begin{equation}\label{MAL reversible equation}
\frac{d c}{dt}=\sum_{\rho}\gamma_{\rho} r_{\rho},
\end{equation}
where $c$ is the vector of concentrations with coordinates $c_i$ and $\gamma_{\rho}$ is the stoichiometric vector of the elementary reaction, $\gamma_{\rho i}=\beta_{\rho i}-\alpha_{\rho i}$ (gain minus loss).

The principle of detailed balance for the MAL kinetics means that $k^{\pm}_{\rho}>0$ and
there exists a positive point of detailed balance $c^*$ ($c^*_i>0$), where
\begin{equation}\label{MAL detailed balance}
r^+_{\rho}(c^*)=r^-_{\rho}(c^*) \,(=r^*), \mbox{ i.e. }k^+_{\rho}\prod_{i=1}^n (c_i^*)^{\alpha_{\rho i}}=k^-_{\rho}\prod_{i=1}^n (c_i^*)^{\beta_{\rho i}}=r^*\, .
\end{equation}
For a given positive point of detailed balance, $c^*$, the reaction rates include $m$
independent positive constants, equilibrium fluxes $r^*_{\rho}$, instead of $2m$ rate
constants $k^{\pm}_{\rho}$:
\begin{equation}\label{MALdbReactRates}
r^+_{\rho}(c)=r^*_{\rho}\prod_{i=1}^n \left(\frac{c_i}{c_i^*}\right)^{\alpha_{\rho i}},
\;\; r^-_{\rho}(c)=r^*_{\rho}\prod_{i=1}^n \left(\frac{c_i}{c_i^*}\right)^{\beta_{\rho
i}}.
\end{equation}
$H$-theorem for MAL kinetics with detailed balance is similar to Boltzmann's $H$-theorem.
Take
\begin{equation}\label{H-chemical}
H(c)=\sum_i c_i \left(\ln\left(\frac{c_i}{c_i^*}\right)- 1\right).
\end{equation}
Simple calculation gives that for the kinetic equations (\ref{MAL reversible equation})
with reaction rate functions (\ref{MALdbReactRates})
\begin{equation}\label{ChemicalEntropyProduction}
\frac{d H}{dt}=-\sum_{\rho}(r^+_{\rho}(c)-r^-_{\rho}(c))(\ln r^+_{\rho}(c)-\ln
r^-_{\rho}(c))\leq 0
\end{equation}
and $$\frac{d H}{dt}=0\mbox{ if and only if }r^+_{\rho}(c)=r^-_{\rho}(c) \mbox{ for all }
\rho $$ because $(x-y)(\ln x - \ln y)\geq 0$ for all positive $x,y$ and it is zero if and
only if $x=y$. Hence, if there exists a positive point of detailed balance than $H(c)$
decreases monotonically in time and all the equilibria are the points of detailed balance
\cite{VolpertKhudyaev1985}.

Physically, the constructed equations correspond to chemical reactions in a system with
constant volume and temperature. For other classical conditions (isobaric systems,
isolated systems, etc, the Lyapunov functionals are also known (see, for example,
\cite{Hangos2010,Yablonskiiatal1991}).

For many real  systems the reaction mechanism includes both reversible and irreversible
reactions. For them some reverse reactions are absent in the reaction mechanism
(\ref{reversible mechanism}). (It is convenient to use such notations that all direct
reactions are present and some reverse reactions are absent). The systems with
irreversible reactions which are the limits of the fully reversible systems with detailed
balance when some of the equilibrium concentrations tend to zero are described
\cite{GorbanMirkesYablonsky2013,GorbYabCES2012}. If the reversible systems obey the
principle of detailed balance then the limit system with some irreversible reactions must
satisfy the {\em extended principle of detailed balance}. It is proven in the form of two
conditions: (i) the reversible part satisfies the principle of detailed balance and (ii)
the convex hull of the stoichiometric vectors of the irreversible reactions does not
intersect the linear span of the stoichiometric vectors of the reversible reactions.
These conditions imply the existence of the global Lyapunov functionals and alow an
algebraic description of the limit behavior. The extended principle of detailed balance
is closely related to the Grigoriev -- Milman -- Nash theory of binomial varieties
\cite{GrigorievMilman2012}.

\subsection{Thermodynamics beyond detailed balance}

In the original form of $H$-theorem the microscopic reversibility (invariance of the
microscopic description with respect to time reversal) is used to prove the macroscopic
irreversibility, the existence of the time arrow ($H$ decreases monotonically due to
kinetic equations). Elegant paradoxical form of this reasoning leaves, nevertheless,
concern about its generality: does the macroscopic irreversibility need the microscopic
reversibility? In 1887 Lorentz formulated this concern explicitly. He stated that the
collisions of polyatomic molecules are irreversible and, therefore, Boltzmann's
$H$-theorem is not applicable to the polyatomic media \cite{Lorentz1887}. Boltzmann found
the solution immediately and invented what we call now semidetailed balance or cyclic
balance or complex balance \cite{Boltzmann1887}. For the Boltzmann equation this new
condition allows a nice schematic representation (see Figure~\ref{scheme1} for detailed
balance and Figure~\ref{ChemeComBal} for complex balance). Now, it is proven that the
Lorentz objections were wrong and the detailed balance conditions hold for polyatomic
molecules \cite{CercignaniLampis1981}. Nevertheless, this discussion was seminal and
stimulated Boltzmann to discover new general conditions of thermodynamic behavior.
\begin{figure}[h]
\centering{
\includegraphics[height=0.12\textwidth]{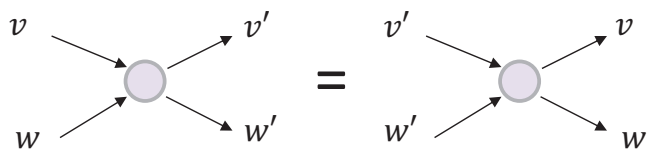}
}\caption{Schematic representation of detailed balance for collisions. The four-tail
scheme represents intensity of the equilibrium flux of collisions with given
velocities.\label{scheme1}}
\end{figure}
\begin{figure}[h]
\centering{
\includegraphics[height=0.12 \textwidth]{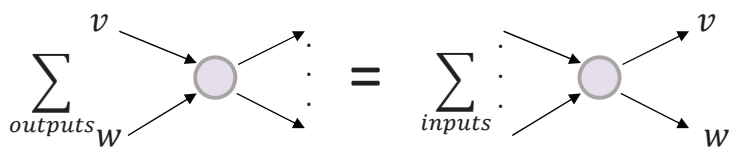}
}\caption{\label{ChemeComBal}Boltzmann's  cyclic  balance is a summarised detailed
balance condition: at equilibrium the sum of intensities of collisions with a given input
$v+w\to \ldots$ coincides with the sum (or integral) of intensities of collisions with
the same output $\ldots \to v+w$.}
\end{figure}

For the MAL kinetics, the Boltzmann cyclic balance condition was rediscovered in 1972
\cite{HornJackson1972}. It got the name ``complex balance'' (balance of complexes). The
complex balance condition has the form of summarised detailed balance
(Figure~\ref{ChemeComBal}).

Consider a reaction mechanism in the form
\begin{equation}\label{irreversible mechanism}
\alpha_{\rho 1}A_1+ \ldots + \alpha_{\rho n} A_n \to \beta_{\rho 1}A_1+ \ldots +
\beta_{\rho n} A_n \, .
\end{equation}
Here, the reverse reactions if they exist participate separately from the direct
reactions and reversibility is not compulsory. This form is convenient for systems
without detailed balance. The MAL reaction rate is $r_{\rho}=k_{\rho}\prod_i
c_i^{\alpha_i}$ and the kinetic equations have again the form (\ref{MAL reversible
equation}).

A positive concentration vector $c^*$ is an equilibrium if $\sum_{\rho}\gamma_{\rho} r^*_{\rho}=0$, where
$r_{\rho}=k_{\rho}\prod_i (c_i^*)^{\alpha_{\rho i}}$. Complexes are the formal sums in the left and right hand sides of
(\ref{irreversible mechanism}). There are $2m$ vectors of coefficients $\alpha_{\rho}=(\alpha_{\rho i})$ and $\beta_{\rho}=(\beta_{\rho i})$
($\rho=1,\ldots, m$). Some of them might coincide. Let $\{y_1, \ldots, y_q\}$ be the distinct coefficient vectors:
for each $y_j$ there exists such $\rho$ that $y_j=\alpha_{\rho}$ or $y_j=\beta_{\rho}$, and for each $\rho $,
$\alpha_{\rho}$ there exists such $j,l$ that $y_j=\alpha_{\rho}$  and $y_l=\beta_{\rho}$.

A positive point $c^*$ is a point of complex balance if for each $y_j$
\begin{equation}\label{complex_balance}
\sum_{\rho, y_j=\alpha_{\rho}}r^*_{\rho}=\sum_{\rho, y_j=\beta_{\rho}}r^*_{\rho} \;\; (j=1, \ldots, q).
\end{equation}
This is exactly the summarized detailed balance condition (compare it to
Figure~\ref{ChemeComBal}). The complex balance conditions (\ref{complex_balance}) are
sufficient for the $H$-theorem: $dH/dt \leq 0$. To demonstrate this inequality, we
consider the deformed stoichiometric mechanism with the stoichiometric vectors which
depend on parameter $\lambda\in[0,1]$:
$$\acute{\alpha}_{\rho}(\lambda)=\lambda \alpha_{\rho}+(1-\lambda)\beta_{\rho}\, , \;
\acute{\beta}_{\rho}(\lambda)=\lambda \beta_{\rho}+(1-\lambda)\alpha_{\rho}.$$ Introduce
an auxiliary function $\theta(\lambda)$, that is the sum of the reaction rates of the
deformed mechanism (with the same equilibrium fluxes). For a given concentration vector
$c$
$$\theta(\lambda)=\sum_{\rho} r^*_{\rho}\prod_{i=1}^n \left(\frac{c_i}{c_i^*}\right)^{\acute{\alpha} (\lambda)_{\rho i}}. $$
Simple calculation gives
$$ \frac{d H}{dt}=-\left.\frac{d\theta}{d \lambda}\right|_{\lambda=1} .$$
The function $\theta(\lambda)$ is convex and the complex balance conditions
(\ref{complex_balance}) imply $\theta(0)=\theta(1)$, therefore under these conditions
$\theta'(1)\geq 0$ and $H$ monotonically decreases in time.

For first order kinetics (continuous time Markov chains or master equation) the complex
balance conditions are just the stationarity conditions (the so-called balance equations)
and hold at every positive equilibrium. This gives immediately the $H$-theorem for first
order kinetics with positive equilibrium $c^*$ and without any additional conditions.

\subsection{From Markov kinetics to MAL with complex balance condition}

The semidetailed balance conditions (Figure~\ref{ChemeComBal}) for Boltzmann's equation
were produced by Stueckelberg \cite{Stueckelberg1952} from the Markov model of collisions
(Stueckelberg used the $S$-matrix notations and presented the balance equation as the
unitarity condition).

MAL for catalytic reactions with a priori unknown kinetic law was obtained in the famous
work of Michaelis and Menten \cite{MichaelisMenten1913}. They postulated that substrates
form complexes (`compounds') with enzymes, these compounds are in equilibrium with
enzymes (fast equilibria), and the concentrations of the compounds is small. The
compound-substrates equilibria can be described by equilibrium thermodynamics and the
kinetics of the compounds transformations is just a Markov chain (because for very small
concentrations of reagents only first order reactions survive). These asymptotic
assumptions lead to MAL. Michaelis and Menten studied very simple reaction, therefore the
additional relations between reaction rate constants did not appear but the Stueckelberg
approach extended to the general reaction kinetics gives the semidetailed balance
(complex balance) condition \cite{GorbanShahzad2011}.

The asymptotic assumptions of the Michaelis--Menten--Stueckelberg works are illustrated
by Figure~\ref{MMSlimit}. Each complex ($\sum_i \alpha_{\rho i}A_i$ or $\sum_i
\beta_{\rho i}A_i$) is associated with its compound ($B_{\rho}^{\pm}$ in Figure). It is
assumed that the complex is in fast equilibrium with its compound and the concentration
of compounds can be found by conditional minimization of a thermodynamic potential (under
isothermal isochoric condition this is the free energy). The second asymptotic assumption
means that the concentrations of compounds are small with respect to the concentration of
reagents. This condition allows us to find the concentrations of different compounds
independently. For the perfect free energy $F=constant\times H$ with $H$ given by
(\ref{H-chemical}) these concentrations might be found explicitly.

It should be stressed that the fast equilibrium assumption was later eliminated from
enzyme kinetics by Briggs and Haldane \cite{BriggsHaldane1925} and what is often called
the Michaelis--Menten kinetics is the Briggs--Haldane kinetics (for the modern analysis
of this system we refer to \cite{Segel89}). Nevertheless, the idea of intermediate
complexes which are in fast equilibria with stable reagents is crucially important for
production of dynamic MAL from thermodynamics.  This idea was reanimated and
systematically used in the theory of the activated complex and reaction rates
\cite{Eyring1935,Eyring1936,LaidlerTweedale2007}. Therefore, the
Michaelis--Menten--Stueckelberg limit (Figure~\ref{MMSlimit}) may be called the
Michaelis--Menten--Stueckelberg--Eyring limit.

\begin{figure}[h]
\centering{
\includegraphics[width=0.5 \textwidth]{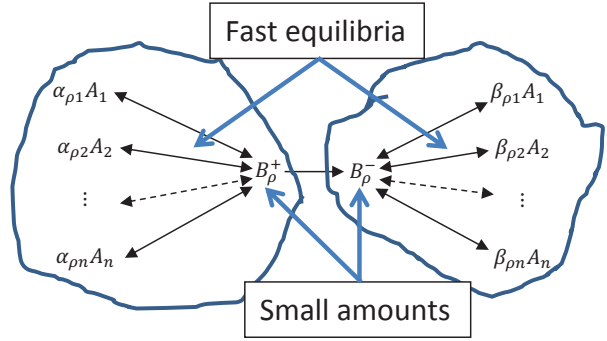}
}\caption{\label{MMSlimit}The Michaelis--Menten--Stueckelberg limit.}
\end{figure}

In our work, we study the Michaelis--Menten--Stueckelberg limit of Markov processes with
general space of states and obtain for them the {\em generalized MAL (GMAL)} with
thermodynamic properties.

\section{Derivation of MAL and GMAL on the arbitrary state space}

\subsection{Thermodynamics of particles}

To set a scenery suppose a species or a particle can be represented by a point $x$ in a locally compact metric space $X$
with some fixed Radon measure $M(dx)$. The distribution of (possibly infinitely many) particles in $X$
can be specified by a finite measure. We shall deal only with distributions that have densities
(concentrations) $c(x) \in L^1(M)$ with respect to $M$.

Let the thermodynamic properties of a concentration $c$ be characterized by the 'free
energy', which is given generally by a functional $F(c(.))$ defined on $L^1(M)$ (or some
its subspace, the domain of $F$). We shall assume that $F$ is smooth in the sense that
the variational derivative $\de F/\de c$ (for positive $c$) exists with respect to $M$,
defined by the equation
\[
\left.\frac{d}{dh}\right|_{h=0} F(c+h\om)=\int_X \frac{\de F}{\de c(x)}\om (x) M(dx)
\]
for $c,\om$ from the domain of $F$.

For many applications, $\de F/\de c$ has a logarithmic singularity as $c\to 0$. Further
on we assume the positivity of $c$ when it is necessary.

Two basic examples that cover all known physical models should be kept in mind. In the first one
\begin{equation}
\label{eqfreeengen1}
F(c(.))=\int_X \psi (c(x)) M(dx)
\end{equation}
with a function $\psi$, which is smooth on positive arguments. In this case
\[
\frac{\de F}{\de c(x)}=\psi ' (c(x)).
\]
This includes the case of finite $X=\{1, \cdots , k\}$ with
\[
F(c_1, \cdots , c_k)=\sum_j \psi_j(c_j).
\]

As another particular case let us mention the standard perfect gas free energy
given by
\begin{equation}
\label{eqfreeengen11} F(c(.))=\int_{\R^d} c(x) \left(\ln \frac{c(x)}{c^*(x)}-1\right) dx,
\quad \frac{\de F}{\de c(x)}=\ln \frac{c(x)}{c^*(x)}
\end{equation}
with some equilibrium distribution $c^*$ (we omit here the constant factors).

Of interest is also its more general version,
so-called 'generalized entropy' function (with inverted sign)
\begin{equation}
\label{eqdefgenentfu}
H_h(c \| c^*)=\int  c^*(x) h\left( \frac{c (x)}{c^*(x)}\right) M(dx),
\end{equation}
where $h$ is any convex smooth function on $\R_+=\{x>0\}$. Function \eqref{eqfreeengen11} is obtained from \eqref{eqdefgenentfu}
for $h(x)=x(\ln x-1)$. Choosing $h(x)=-\ln x$ leads to the so-called Burg relative entropy
\begin{equation}
\label{eqdefgenentfuBurg}
H_h(c \| c^*)=-\int  c^*(x) \ln \frac{c (x)}{c^*(x)} M(dx).
\end{equation}

In the second example
$X=\R^d \times \{1, \cdots , k\}$ with $M(dx)$ being Lebesgue measure on each component (more generally,
instead of $\R^d$ one can use a manifold, but we shall stick to $\R^d$ for simplicity). With some abuse of
notation we shall denote the elements of $X$ by a pair $(x, j)$, $x\in \R^d$, $j=1, \cdots, k$ or sometimes by $x_j$.
The concentration $c$ becomes a vector $c=(c_1(x), \cdots , c_k(x))$ with its gradient
$\nabla c(x)=(\nabla c_1(x), \cdots , \nabla c_k(x))$, where
\[
\nabla c_j(x)=\left(\frac{\pa c_j}{\pa x^1}, \cdots , \frac{\pa c_j}{\pa x^d}\right).
\]
The free energy is specified by the equation
\begin{equation}
\label{eqfreeengen2}
F(c(.)) =\int \psi (c(x),\nabla c(x)) \, dx
\end{equation}
with some smooth function $\psi$.
This includes the case of finite $X =\{1, \cdots , k\}$ with
\begin{equation}
\label{eqfreeengen3}
F(c(.)) =\psi (c_1,\cdots , c_k).
\end{equation}
The well established particular case of \eqref{eqfreeengen2} is
\begin{equation}
\label{eqfreeengen4}
F(c(.)) =\int \left(\psi (c(x))+\frac12 |\nabla c(x)|^2\right) \, dx,
\end{equation}
used in the Chan-Hilliard model of diffusion.
The simplest version of \eqref{eqfreeengen2} is the decomposable case:
\begin{equation}
\label{eqfreeengen5}
F(c(.)) =\int \sum_{j=1}^k \psi_j (c_j(x),\nabla c_j(x)) \, dx.
\end{equation}

The variational derivative for $F$ of type \eqref{eqfreeengen2} is the standard Euler-Lagrange one:
\begin{equation}
\label{eqfreeengen6}
\frac{\de F}{\de c_j(x)}=\frac{\pa \psi}{\pa c_j}-\sum_l \frac{\pa}{\pa x^l} \frac{\pa \psi}{\pa \nabla_l c_j}.
\end{equation}

\subsection{Compounds}

Our objective is to describe the process of transformation of particles (chemical reactions, collisions, etc).
The collections of $k$ particles can be given by points in $SX^k$ and the collection of an arbitrary number of
particles in $S\XC=\cup_{j=1}^{\infty} X^j$, where $SX^k$ and $S\XC$ are the quotient-spaces
of $X^k$ and $\XC=\cup_{j=1}^{\infty} X^j$ with respect to all permutations. Symmetrical probability laws on
$X^k$ (which are uniquely defined by their projections to $SX^k$)
are called exchangeable systems of $k$ particles. With some abuse of notations we shall use the same bold face letter notations,
say $\x=(x_1, \cdots, x_k)$, both to denote the points of $X^k$ and $SX^k$.
We shall also use the notation $\XC_{\ge k}=\cup_{j\ge k}X^j$.

The main idea of the {\it intermediate state} or the {\it activation complex} assumption (that we shall adopt here)
is that any reaction changing the collection of $k$ particles $\x=(x_1, \cdots, x_k)$ to the collection
of $l$ particles $\y=(y_1, \cdots, y_k)$ is not a one step operation, but the three step one:
before the interaction is enabled the $k$ particle $\x=(x_1, \cdots, x_k)$ should form the intermediate state
$\bar \x$, which we shall call a {\it compound} of size $k$ (consisting of the same $k$ particles
$\x$), then the compound $\bar \x$ turns to the compound $\bar \y$,
which in turn can be dissolved into its components $\y$:
\begin{equation}
\label{eqcompoundreac}
\x=(x_1 \cdots , x_k) \to \bar \x \to \bar \y \to \y=(y_1, \cdots , y_l).
\end{equation}

This concept of the intermediate states allows one to speak about the distribution of compounds present
in the system. Introducing some fixed symmetric measures $M_k$ on $SX^k$ allows one to reduce attention
to distributions specified by the densities (concentrations), which, for compounds of size $k$,
are given by the symmetric functions $\zeta_k (\x) \in L^1(M_k)$.

Let us denote by $\M$ the measure on $\XC$ with the coordinates $(M=M_1, M_2, \cdots)$
and by $\mu(x, d\y)$ the stochastic kernel on $\XC$ with the coordinates $\mu_k(x,d\y)$ so that
for a function $f=(f_1,f_2, \cdots )$ on $\XC$,
\[
\int_{S\XC} f (\x) \M(d\x)=\sum_{j=1}^{\infty} \int_{SX^j} f_j(\x) M_j (d\x),
\quad \int_{S\XC} f (x,\y) \mu(x,d\y)=\sum_{j=1}^{\infty} \int_{SX^j} f_j(x,\y) \mu_j (x,d\y).
\]
The supports of measures $M_j$ specify the set of compounds that can be formed in the system.

Of course, the simplest example of the measures $M_k$ are the projections on $SX^k$ of the products
 $M(dx) \otimes \cdots \otimes M(dx)$
($k$ times), but this example is not sufficient (as we shall see below) to cover all cases of interest.
However, in order to develop a theory, some link between $M_k$ and $M$ should be made.
Our main assumption about $M_k$ will be that the projection of all $M_k$ on the one-particle states
is absolutely continuous with respect to $M$, namely
\begin{equation}
\label{eqasscompbasmaes}
M_k(dx_1 \cdots dx_k) =\frac1k M(dx_1) \mu_k(x_1, dx_2 \cdots dx_k),
\end{equation}
with a symmetric stochastic kernel $\mu_k(x, d\x)$, $x\in X, \x \in SX^{k-1}$.
  By symmetry, \eqref{eqasscompbasmaes} rewrites as
\begin{equation}
\label{eqasscompbasmaes1}
M_k(dx_1 \cdots dx_k) =\frac1k M(dx_j) \mu_k(x_j, dx_2 \cdots dx_{j-1} dx_{j+1} \cdots dx_k),
\end{equation}
for any $j$.

This assumption is crucial for the possibility to relate the concentration of particles with the concentration of compounds.
Namely, as shown in Appendix (see Proposition \ref{kernstoich}),
the kernels $\mu$ in \eqref{eqasscompbasmaes} can be chosen in such a way that
if $\zeta_k (x_1, \cdots , x_k)=\zeta_k (\x)$ is the concentration of the compounds
$\bar \x$ of size $k$, the concentration of particles involved in these compounds equals
\begin{equation}
\label{eqasscompbasmaes2}
\int_{SX^{k-1}} \zeta_k (x, x_2, \cdots , x_k) \mu_k(x, dx_2 \cdots dx_k).
\end{equation}

Notice that the coefficient $1/k$ was introduced in \eqref{eqasscompbasmaes}
to avoid any coefficients in \eqref{eqasscompbasmaes2}.
The kernels $\mu_k$ appearing in Proposition \ref{kernstoich}
will be called {\it stoichiometric kernels}, as they present natural analogs
of the stoichiometric coefficients of the theory of chemical reactions
on a finite state space.

Therefore, if $\zeta(x)$ is the concentration of free particles (not involved in the compounds), the total
concentration of particles is
\begin{equation}
\label{eqasscompbasmaes3}
c(x) =\zeta (x) + \sum_{k=2}^{\infty}\int_{SX^{k-1}} \zeta (x, x_2, \cdots , x_k) \mu_k(x, dx_2 \cdots dx_k)
=\zeta (x)+\int_{S\XC} \zeta (x, \y)\mu (x, d\y),
\end{equation}
where $\zeta(\x)=\zeta_k(\x)$ for $\x\in X^k$
and $\mu(x, d\y)$ is the stochastic kernel on $\XC$ (stoichiometric kernel) with the 'coordinates' $\mu_k(x,d\y)$ on $SX^k$.

Further on we shall assume for simplicity that the size of possible compounds is uniformly bounded,
so that all sums over sizes used below are finite. This restriction is also natural from the practical point of view,
as the sizes of compounds met in practice are very small
(usually $2$ or at most, and rarely, $3$).

\subsection{QE and QSS}

The {\it quasi-steady-state} (QSS) assumption states that the compounds exist in very small
 concentrations as compared with the concentration of free particle
(because they form and dissolve very quickly) and the {\it quasi-equilibrium} (QE) assumption states
that the reaction of equilibration
between particles and compounds is much faster than the reaction between compounds meaning
that the compounds exist all the time in a fast equilibrium with the set of basic particles.
Let us discuss the important conclusions from these assumptions.

First of all, by QSS, the free energy of the compounds can be taken in the form of the perfect free energy
(the free energy of the ideal gas or of dilute solutions), so that the total free energy of the system becomes
\[
F_{tot}(c, \zeta_2, \zeta_3, \cdots)=F_{tot}(\zeta)
=F(c (.)) + \int_{S\XC_{\ge 2}} \zeta(\x) \left( \ln \frac{\zeta(\x)}{\zeta^*(\x)}-1\right)\M(d\x)
\]
\begin{equation}
\label{eqfreeengenmixcomp}
=F(c (.)) +\sum_{j\ge 2} \int_{SX^j} \zeta_j(\x) \left( \ln \frac{\zeta_j(\x)}{\zeta^*_j(\x)}-1\right)M_j(d\x)
\end{equation}
with $F$ the free energy of particles as introduced above and with $\zeta_j^*$, $j>1$, some equilibrium concentrations.
Generally speaking, $\zeta^*$ should depend on $c$, but again by QSS, the concentration of particles
are large and vary slowly as compared with the compounds implying that this dependence can be neglected in the first approximation.
In the same approximation we do not distinguish the concentrations
 of free particles $\zeta_1(x)$ and their total concentration $c(x)$.

By QE the compounds are all the time in equilibrium with particles. As in equilibrium the free energy takes its minimum,
$\zeta_j$, $j\ge 2$, can be found from the condition of the extremum:
\[
\left.\frac{d}{d\ep}\right|_{\ep=0} \left[F\left(c(.)-\ep \int_{S\XC} \om (., \y)\mu(., d\y)\right)
+\int_{S\XC_{\ge 2}} (\zeta+\ep \om) (\y))  \left( \ln \frac{(\zeta+\ep \om)(\y)}{\zeta^*(\y)}-1\right)\M(d\y)\right]=0,
\]
that should hold for all symmetric functions $\om (\x)$ on $\XC_{\ge 2}$.
By the definition of the variational derivative this implies
\begin{equation}
\label{eqfreeengenmixcompextr}
-\int_X \frac{\de F}{\de c(x)}\int_{S\XC} \om (x, \y)\mu(x, d\y)M(dx)+\int_{S\XC_{\ge 2}} \om (\x)\ln \frac{\zeta(\x)}{\zeta^*(\x)} \M(d\x)=0
\end{equation}
for all $\om$ and hence, by \eqref{eqasscompbasmaes},
\begin{equation}
\label{eqfreeengenmixcompextr1}
\int_{SX^j} \om _j(x_1, x_2, \cdots, x_j)\left(\ln \frac{\zeta_j(\x)}{\zeta^*_j(\x)}-j\frac{\de F}{\de c(x_1)}\right) M_j(dx_1 dx_2 \cdots dx_j)=0
\end{equation}
for all $j\ge 2$ and $\om_j$. By the symmetry this implies
\begin{equation}
\label{eqfreeengenmixcompextr2}
\int_{SX^j} \om _j(x_1,\cdots, x_j)\left(\ln \frac{\zeta_j(\x)}{\zeta^*_j(\x)}
-\sum_{l=1}^j \frac{\de F}{\de c(x_l)}\right) M_j(dx_1 \cdots dx_j)=0.
\end{equation}
Consequently
\begin{equation}
\label{eqfreeengenmixcompextr3}
\ln \frac{\zeta_j(\x)}{\zeta^*_j(\x)}
=\sum_{l=1}^j \frac{\de F}{\de c(x_l)},
\end{equation}
so that finally, for any $j>1$, the minimizing concentrations are
\begin{equation}
\label{eqfreeengenmixcompextr4} \zeta_j(\x;c)=\zeta^*_j(\x) \exp \left\{ \sum_{l=1}^j
\frac{\de F}{\de c(x_l)} \right\}.
\end{equation}

\subsection{Dynamics}

The next consequence of the QSS is that the dynamics of compounds should be linear, because,
their concentration being small, one can neglect their interaction.  This includes the dynamics
of free particles, as their interaction has been accounted for by the formation of compounds.

By \eqref{eqMarkjumpgenmul3}, assuming \eqref{eqdisintmes1} and using the notations for the concentrations of compounds introduced above,
the general jump-type Markov evolution on concentrations $\zeta$ can be written in the concise form

\begin{equation}
\label{eqMarkjumpdyncomp1}
\dot \zeta (\x)=\int_{S\XC} [\zeta(\y)\tilde \nu (\x, d\y) -\zeta (\x) \nu (\x, d\y)].
\end{equation}
Here $\nu$ is some collection of transition kernels
\[
\nu(\x, d\y)=\{ \nu_{k\to l}(\x, d\y), \, \x \in SX^k, \, \y \in SX^l\}
\]
and the collection $\tilde \nu$ is defined via the following equality of measures on each pair $SX^k\times SX^l$:
\begin{equation}
\label{eqMarkjumpdyncomp2}
\nu_{k\to l}(\y, d\x) M_k(d\y)=\tilde \nu_{l\to k}(\x, d\y)M_l (d\x).
\end{equation}

The additional constraint arising from thermodynamics comes from our assumption,
see \eqref{eqfreeengenmixcomp},  that $\zeta^*$ are equilibrium concentrations of compounds
and therefore they should supply equilibrium for their linear evolution \eqref{eqMarkjumpdyncomp1}, that is
\begin{equation}
\label{eqMarkjumpdyncomp1equ}
\int_{S\XC} [\zeta^*(\y)\tilde \nu (\x, d\y) -\zeta^* (\x) \nu (\x, d\y)]=0
\end{equation}
for any $\x$.

It is useful to distinguish a subclass of processes where particles themselves do not serve as compounds, or in other words,
direct transitions $X\to S\XC_{\ge 2}$ and $S\XC_{\ge 2} \to X$ are not allowed:
\begin{equation}
\label{eqcondpartnotcomp}
\nu (x, d\y)=0, \quad \nu (\y,dx)=0, \quad x\in X, \y\in S\XC_{\ge 2}.
\end{equation}
If this is the case condition  \eqref{eqMarkjumpdyncomp1equ} should be understood as
\begin{equation}
\label{eqMarkjumpdyncomp2equ}
\int_{S\XC_{\ge 2}} [\zeta^*(\y)\tilde \nu (\x, d\y) -\zeta^* (\x) \nu (\x, d\y)]=0, \quad \x\in S\XC_{\ge 2}.
\end{equation}
Otherwise, for \eqref{eqMarkjumpdyncomp1equ} to make sense, equilibrium quantities $\zeta^*(x)$, $x\in X$,
should be defined somehow to complement the definitions of $\zeta^*(\x;c)$ for $\x\in S\XC_{\ge 2}$ given by
\eqref{eqfreeengenmixcompextr3}.

A simpler subclass of processes worth being mentioned present the evolutions preserving the number of particles in the compounds.
In this case, the evolution \eqref{eqMarkjumpdyncomp1} decomposes into the independent evolutions in each $SX^k$:
\begin{equation}
\label{eqMarkjumpdyncomp3}
\dot \zeta_k (\x)=\int_{SX^k} [\zeta_k(\y)\tilde \nu_{k\to k} (\x, d\y) -\zeta_k(\x)\nu_{k\to k} (\x,d\y)], \quad k=1,2, \cdots.
\end{equation}

By \eqref{eqasscompbasmaes3}, the evolution of the compounds \eqref{eqMarkjumpdyncomp1}
implies the following evolution of the total concentration $c$:

\begin{equation}
\begin{split}
\dot c(x) =&\int_{S\XC} [\zeta (\y) \tilde \nu (x,d\y) -\zeta (x) \nu (x, d\y)]\\
\label{eqMarkjumpdyncomp4} &+\int_{S\XC} \mu (x, d\y) \int_{S\XC}[\zeta (\z) \tilde \nu
(x,\y, d\z) -\zeta (x,\y) \nu (x, \y, d\z)].
\end{split}
\end{equation}

It remains now to put it all together.
As shown above, by the QE assumptions $\zeta_k$ for $k>1$ are expressed by \eqref{eqfreeengenmixcompextr3} in terms of $c$.
By QSS, for $k=1$ we have approximately
\begin{equation}
\label{eqMarkjumpdyncomp5}
\zeta_1(x)=\zeta_1(x;c)=c(x).
\end{equation}
Finally, apart from the transformation of particles it is natural to allow additionally their movement in $X$ according to some Markov
process with the generator $L$ (only free particles are moving, as the movement of the short-lived compounds can be neglected).
Then the final evolution of the concentration becomes
\begin{equation}
\begin{split}
\dot c(x) =&L^* c(x)+ \int_{S\XC} [\zeta (\y;c) \tilde \nu (x,d\y) -\zeta_1 (x;c) \nu (x,
d\y)] \\ \label{eqMarkjumpdyncomp6} &+\int_{S\XC} \mu (x, d\y) \int_{S\XC}[\zeta (\z;c)
\tilde \nu (x,\y, d\z) -\zeta (x,\y;c) \nu (x, \y, d\z)]
\end{split}
\end{equation}
supplemented by \eqref{eqfreeengenmixcompextr3} and \eqref{eqMarkjumpdyncomp5}. This
evolution can be called the {\it generalized mass action law (GMAL)}. It is the extension
to an arbitrary state space $X$ of the finite-state-space GMAL. The latter was developed
in general by Gorban et al \cite{G1,GorbanBykYabEssays1986,GorbanShahzad2011} following
the ideas of Michaelis-Menten, Eyring, Stueckelberg, and many others. For the diffusion
equations the formalism of GMAL was also elaborated \cite{GorbanSarkisyan2011}.

If condition \eqref{eqcondpartnotcomp} holds (particles are not compounds), evolution \eqref{eqMarkjumpdyncomp6}
rewrites in a simpler form
\begin{equation}
\label{eqMarkjumpdyncomp61}
\dot c(x) =L^* c(x)+\int_{S\XC} \mu (x, d\y) \int_{S\XC_{\ge 2}}[\zeta (\z;c) \tilde \nu (x,\y, d\z) -\zeta (x,\y;c) \nu (x, \y, d\z)].
\end{equation}

If the particles themselves are given in small concentrations, so that their free energy
has the perfect form \eqref{eqfreeengen11}, evolution \eqref{eqMarkjumpdyncomp6} turns to
\begin{equation}
\begin{split}
\dot c(x) =&L^* c(x)+ \int_{S\XC} \left[\c(\y)\frac{\zeta^*(\y)}{\c^*(\y)}  \tilde \nu
(x,d\y) -\c(x)\frac{\zeta^*(x)}{\c^*(x)} \nu (x, d\y)\right] \\
\label{eqMarkjumpdyncomp7} &+\int_{S\XC} \mu (x, d\y)
\int_{S\XC}\left[\c(\z)\frac{\zeta^*(\z)}{\c^*(\z)} \tilde \nu (x,\y, d\z)
-\c(x,\y)\frac{\zeta^*(x,\y)}{\c^*(x,\y)} \nu (x, \y, d\z)\right],
\end{split}
\end{equation}
where $\zeta^*(x)=c^*(x)$ for $x\in X$ and
\[
\c(x_1,\cdots, x_k)=c(x_1) \cdots c(x_k), \quad \c^*(x_1,\cdots, x_k)=c^*(x_1) \cdots c^*(x_k).
\]
This is the evolution of the MAL for an arbitrary state space $X$.

Important to observe that, for the MAL evolution \eqref{eqMarkjumpdyncomp7}, equilibrium
quantities  $\zeta^*(x)=c^*(x)$ for $x\in X$ are explicitly specified from the expression of free energy, and thus the
condition \eqref{eqMarkjumpdyncomp1equ} is well defined without the restriction
\eqref{eqcondpartnotcomp}. Moreover, condition
\eqref{eqMarkjumpdyncomp1equ} supplemented by the similar condition on the free evolution of $c$, that is assuming $L^*c^*=0$, implies
that $c^*$ are equilibrium concentrations to \eqref{eqMarkjumpdyncomp7}.

\subsection{Basic examples}

For the case of a discrete state space $X=\{A_1, \cdots, A_k\}$ it is more convenient (and well established in the literature)
to use separate enumeration of compounds and reactions.  For each reaction $r$, denoting by $\al_{ri}$ the number
of particles $A_i$ entering the input compound $B_r^-$ and by $\be_{ri}$ the number
of particles $A_i$ entering the output compound $B_r^+$, the reaction can be described schematically as
\[
\sum_i \al_{ri} A_i \rightleftharpoons B_r^- \to B_r^+ \rightleftharpoons \sum_i \be_{ri} A_i.
\]
Here $\al_{ri}, \be_{ri}$ are known as stoichiometric coefficients and the vector
$\nu_{r}=(\be_{ri}-\al_{ri})$ as the stoichiometric vector of the reaction $r$.
In this notation the evolution \eqref{eqMarkjumpdyncomp6} becomes
\begin{equation}
\label{eqMarkjumpdyncomp611}
\dot c_i =\sum_{l\neq j} [\kappa_{jl} \zeta_l -\kappa_{lj} \zeta_j] \nu_{ji},
\end{equation}
with some $\kappa_{lj}$ playing the role of transitions $\mu(x,d\y)$ of \eqref{eqMarkjumpdyncomp6};
the summation is over all pairs of compounds $(l,j)$ and
\begin{equation}
\label{eqMarkjumpdyncomp612} \zeta_l=\zeta^*_l\exp \left\{\sum_j \frac{\pa \psi}{\pa
c_j}(c)\nu_{lj} \right\}
\end{equation}
for each compound $l$, with $F,\psi$ from \eqref{eqfreeengen3}.

Most of real life evolutions involve compounds consisting of only two particles. If only pairs to pairs transitions
can occur then GMAL \eqref{eqMarkjumpdyncomp6} and MAL \eqref{eqMarkjumpdyncomp7} take the form
\begin{equation}
\label{eqMarkjumpdyncomp6pair}
\dot c(x) =L^* c(x)
+\int_X \mu (x, dy) \int_{SX^2}[\zeta (z_1,z_2;c) \tilde \nu (x,y, dz_1dz_2) -\zeta (x,y;c) \nu (x,y, dz_1dz_2)],
\end{equation}
and respectively
\begin{equation}
\begin{split}
&\dot c(x) =L^* c(x)\\
 \label{eqMarkjumpdyncomp7pair} &+\int_X \mu (x, dy)
\int_{SX^2}\left[\frac{c(z_1)c(z_2)} {c^*(z_1)c^*(z_2)}\zeta (z_1,z_2) \tilde \nu (x,y,
dz_1dz_2)
 -\frac{c(x)c(y)}{c^*(x)c^*(y)}\zeta (x,y) \nu (x,y, dz_1dz_2)\right].
 \end{split}
\end{equation}

Notice now that the densities entering the kernel $\mu$ can be transferred to the rates $\nu$.
Hence, as was already pointed out, basically only the support of $\mu$ is essential.
It turns out that two particular cases cover all interesting examples. The first comes from the assumption that any pair of particles can interact.
In this case, the measure $M_2$ on pairs can be taken to be proportional to the product measure $M(dx)M(dy)$ and then one can take
\begin{equation}
\label{eqMarkjumpdyncomp8pair}
\mu (x, dy) =M (dy).
\end{equation}
In the second case, the state space $X$ is the product $\R^m\times V$ equipped with the measure $dx M(dv)$, where the first component is interpreted as
the position in space and where it is assumed that a pair of particles can interact if and only if their
positions in space coincide.
In this case one has
\begin{equation}
\label{eqMarkjumpdyncom9pair}
\mu((y,v), d(z,w))=\de (y-z) M(dw).
\end{equation}
For instance, the full Boltzmann equation is of that type.

Let us also distinguish two cases of interest concerning transitions $\nu$. The simplest case
is of course when the transitions $\nu (x_1,x_2; dy_1 dy_2)$ are absolutely continuous with
respect to the product measure $M(dy_1)M(dy_2)$. If also \eqref{eqMarkjumpdyncomp8pair} holds
evolution \eqref{eqMarkjumpdyncomp6pair} turns to
\begin{equation}
\label{eqMarkjumpdyncomp10pair}
\dot c(x) =L^* c(x)
+\int_X M(dy) \int_{SX^2}[\zeta (z_1,z_2;c) \nu (z_1, z_2;x,y) -\zeta (x,y;c) \nu (x,y; z_1, z_2)] M(dz_1)M(dz_2).
\end{equation}
This example is however of rather limited applicability.
More interesting situation occurs when there is given another measure space $\Om$ with the measure $d\om$ and a family
of $M_2$-measure-preserving bijections $G_{\om}:X^2\to X^2$ depending on $\om$ as a parameter such that
\begin{equation}
\label{eqMarkjumpdyncomp11pair}
\nu (\x, d\y) =\int_{\Om} B (\x,\om)\de(\y- G_{\om}(\x)) d\om d\y
\end{equation}
 with some function $B(\x, \om)$ on $X^2\times \Om$. Since
\[
 \int_{X^2}\int_{\Om}  f(\y,\x) \de (\x-G_{\om}(\y))B(\y,\om) d\x d\om M_2(d\y)
 = \int_X\int_{\Om}  f(\y,G_{\om}(\y)) B(\y,\om) d\om M_2(d\y)
\]
\[
=\int_X\int_{\Om}  f(G^{-1}_{\om}(\x),\x) B(G^{-1}_{\om}(\x),\om) d\om M_2(d\x)
= \int_{X^2}\int_{\Om}  f(\y,\x) \de (\y-G^{-1}_{\om}(\x))B(\y,\om) d\y d\om M_2(d\x),
\]
it follows (see \eqref{eqdisintmes1})  that
\[
\tilde \nu (\x, d\y) =\int_{\Om} B(\y,\om)\de(\y- G^{-1}_{\om}(\x)) d\om d\y.
\]
Consequently, assuming again \eqref{eqMarkjumpdyncomp8pair},
evolution \eqref{eqMarkjumpdyncomp6pair} turns to
\begin{equation}
\label{eqMarkjumpdyncomp12pair}
\dot c(x) =L^* c(x)
+\int_X M(dx_2) \int_{\Om}[\zeta (G^{-1}_{\om}(\x);c) B (G^{-1}_{\om}(\x),\om) -\zeta (\x;c) B(\x, \om)] d\om.
\end{equation}
In particular, if $B(x,\om)$ is invariant under the action of $G_{\om}$, this simplifies to
\begin{equation}
\label{eqMarkjumpdyncomp13pair}
\dot c(x) =L^* c(x)
+\int_X M(dx_2) \int_{\Om}B(\x,\om) [\zeta (G^{-1}_{\om}(\x);c) -\zeta (\x;c) ] d\om.
\end{equation}
For instance, the spatially homogeneous Boltzmann equation is of that type,
as well as the mollified Boltzmann equation and their $k$th order extension, see \cite{Ko03}.

To give an example of evolutions arising from \eqref{eqMarkjumpdyncom9pair}, assume
the rates $\nu$ are absolutely continuous with respect to the second variable
and are independent of the first variable. Then
evolution \eqref{eqMarkjumpdyncomp6pair} turns to
\begin{equation}
\begin{split}
\dot c(y_1,v_1) =&L^* c(y_1,v_1)+\int_{\R^d} \int_V  M(dv_2) \de (y_2-y_1)
\int_{SV^2}M(dw_1)M(dw_2)\\
\label{eqMarkjumpdynsingpair1} &\times [\zeta (w_1,y_1;w_2,y_2;c) \nu (w_1, w_2; v_1,v_2)
-\zeta (v_1,y_1;v_2,y_2;c) \nu (v_1,v_2; w_1, w_2)].
\end{split}
\end{equation}

\section{Analysis of equilibria}

\subsection{Evolution of the free energy}

We are interested in conditions ensuring the decrease of the free energy $F(c(.))$.
If $c(x)$ evolves according to \eqref{eqMarkjumpdyncomp61}, the free energy evolves as
(where \eqref{eqasscompbasmaes} is used to get rid of $\mu$)
\begin{equation*}
\begin{split}
\dot F(c(.))= &\int_X \frac{\de F}{\de c(y)}\dot c(y) M(dy)=\int_X \frac{\de F}{\de
c(y)}L^* c(y) M(dy) \\
 &+\sum_{k=2}^{\infty} \int_{SX^k} k M_k(d\y)\frac{\de F}{\de c(y_1)} \int_{S\XC_{\ge 2}}
[\zeta (\z;c) \tilde \nu (\y;d\z) -\zeta (\y;c) \nu (\y, d\z)] \M(d\z).
\end{split}
\end{equation*}
Using symmetry and introducing a handy special notation
\[
\left(\frac{\de F}{\de c(\y)}\right)^{\oplus}=\sum_{j=1}^l \frac{\de F}{\de c(y_j)}, \quad \y=(y_1, \cdots, y_l) \in X^l, \quad l=1,2, \cdots ,
\]
this rewrites as
\begin{equation*}
\begin{split}
\dot F(c(.))= &\int_X \frac{\de F}{\de c(y)}L^* c(y) M(dy)\\
&+\int_{S\XC_{\ge 2}} M(d\y)  \int_{S\XC_{\ge 2}} \left(\frac{\de F}{\de
c(\y)}\right)^{\oplus} [\zeta (\z;c) \tilde \nu (\y;d\z) -\zeta (\y;c) \nu (\y, d\z)]
\end{split}
\end{equation*}
or, using the definition of $\tilde \nu$, as
\begin{equation*}
\begin{split}
\dot F(c(.))=&\int_X \frac{\de F}{\de c(y)}L^* c(y) M(dy)\\
&+\int_{S\XC_{\ge 2}} \int_{S\XC_{\ge 2}}\left(\frac{\de F}{\de c(\y)}\right)^{\oplus}
\left[\zeta (\z;c) \nu (\z, d\y)M(d\z)- \zeta (\y;c) \nu (\y, d\z)M(d\y)\right].
\end{split}
\end{equation*}
Finally, relabeling the variables in the second term of the last integral yields
\begin{equation}
\begin{split}
\dot F(c(.))=&\int_X \frac{\de F}{\de c(x)}L^* c(x) M(dx)\\
\label{eqMarkjumpdynfree1} &+\int_{S\XC_{\ge 2}}  \int_{S\XC_{\ge 2}}
\left[\left(\frac{\de F}{\de c(\y)}\right)^{\oplus}-\left(\frac{\de F}{\de
c(\z)}\right)^{\oplus}\right] \zeta (\z;c) \nu (\z, d\y) \M(d\z).
\end{split}
\end{equation}

Turning to the general case \eqref{eqMarkjumpdyncomp6} we find similarly that
\begin{equation}
\begin{split}
\dot F(c(.))=&\int_X \frac{\de F}{\de c(x)}L^* c(x) M(dx)\\
\label{eqMarkjumpdynfree2} &+\int_{S\XC}  \int_{S\XC} \left[\left(\frac{\de F}{\de
c(\y)}\right)^{\oplus}-\left(\frac{\de F}{\de c(\z)}\right)^{\oplus}\right] \zeta (\z;c)
\nu (\z, d\y)\M(d\z).
\end{split}
\end{equation}

\subsection{Complex balance and detailed balance}

Evolutions \eqref{eqMarkjumpdynfree1} or \eqref{eqMarkjumpdynfree2} can be considered as continuous-state-space analogs
of the discrete state-space representation giving the dynamics of the free energy
in terms of the sum over reactions, as here we have the representation in terms of the integral
over the pairs $(\y,\z)$ that effectively parametrized possible reactions.

With this analogy in mind, and dealing again first with evolution \eqref{eqMarkjumpdyncomp61} and \eqref{eqMarkjumpdynfree1},
  we can now generalize the trick used for the discrete case and introduce the auxiliary function
\begin{equation}
\begin{split}
\theta(\la)=\theta (\la;c) =&\int_{S\XC_{\ge 2}}\int_{S\XC_{\ge 2}} M(d\z) \zeta^* (\z)
\nu (\z;d\y)\\
\label{eqdefGorbanTheta} &\times \exp \left\{ \la \left(\frac{\de F}{\de
c(\z)}\right)^{\oplus} +(1-\la) \left(\frac{\de F}{\de c(\y)}\right)^{\oplus}\right\},
\end{split}
\end{equation}
so that
\begin{equation}
\begin{split}
\theta '(\la)=\frac{d}{d\la} \theta (\la;c)=& \int_{S\XC_{\ge 2}}\int_{S\XC_{\ge 2}}
M(d\z) \zeta^* (\z) \nu (\z;d\y)
 \left[ \left(\frac{\de F}{\de c(\z)}\right)^{\oplus}- \left(\frac{\de F}{\de
 c(\y)}\right)^{\oplus}\right] \\
&\times\exp\left\{ \la  \left(\frac{\de F}{\de c(\z)}\right)^{\oplus} +(1-\la)
\left(\frac{\de F}{\de c(\y)}\right)^{\oplus} \right\}, \\
\theta ''(\la)=\frac{d^2}{d\la^2} \theta (\la;c)= &\int_{S\XC_{\ge 2}}\int_{S\XC_{\ge 2}}
M(d\z) \zeta^* (\z) \nu (\z;d\y)
 \left[ \left(\frac{\de F}{\de c(\z)}\right)^{\oplus}- \left(\frac{\de F}{\de c(\y)}\right)^{\oplus}\right]^2
\label{eqdefGorbanThetasecder} \\
&\times\exp\left\{ \la  \left(\frac{\de F}{\de
c(\z)}\right)^{\oplus} +(1-\la)  \left(\frac{\de F}{\de c(\y)}\right)^{\oplus} \right\}.
\end{split}
\end{equation}

Hence $\theta''(\la)\ge 0$, so that $\theta (\la)$ is a convex function,
and moreover, by \eqref{eqMarkjumpdynpairreac3},
\begin{equation}
\label{eqMarkjumpdynfree3}
\dot F(c(.))=\int_X \frac{\de F}{\de c(x)}L^* c(x) M(dx)-\theta '(1).
\end{equation}

Consequently, the conditions
\begin{equation}
\label{eqMarkjumpdynfree4}
\int_X \frac{\de F}{\de c(x)}L^* c(x) M(dx)\le 0
\end{equation}
and
\begin{equation}
\label{eqMarkjumpdynfree5}
\theta (0)\le \theta (1)
\end{equation}
are sufficient for the decrease of free energy along the evolution \eqref{eqMarkjumpdyncomp10pair}: $\dot F(c(.))\le 0$.

Inequality \eqref{eqMarkjumpdynfree5} introduced in \cite{G1} is referred to as
$G$-inequality. It is a natural weakening of a stronger condition
\begin{equation}
\label{eqMarkjumpdynfree6}
\theta (0)= \theta (1),
\end{equation}
which is often easier to analyze, since it rewrites as
\[
\int_{S\XC_{\ge 2}}\int_{S\XC_{\ge 2}} M(d\z) \zeta^* (\z) \nu (\z;d\y)
\exp \left\{ \left(\frac{\de F}{\de c(\z)}\right)^{\oplus}\right\}
\]
\[
=\int_{S\XC_{\ge 2}}\int_{S\XC_{\ge 2}} M(d\z) \zeta^* (\z) \nu (\z;d\y)
\exp \left\{ \left(\frac{\de F}{\de c(\y)}\right)^{\oplus}\right\}
\]
or equivalently, again using the definition of $\tilde \nu$, as
\begin{equation}
\label{eqMarkjumpdynfree7}
\int_{S\XC_{\ge 2}}\int_{S\XC_{\ge 2}} M(d\z) [\zeta^* (\z) \nu (\z;d\y) -\zeta^* (\y) \tilde \nu (\z;d\y)]
\exp\left\{\left(\frac{\de F}{\de c(\z)}\right)^{\oplus} \right\}=0.
\end{equation}
Assuming the functional $F$ is rich enough, so that the linear combinations of the exponents
\[
\exp\left\{ \left(\frac{\de F}{\de c(\y)}\right)^{\oplus} \right\}
\]
for all continuous functions $c$ are dense in the space of continuous functions on $S\XC_{\ge 2}$,
as is the case for the MAL evolution, \eqref{eqMarkjumpdynfree7} implies
\begin{equation}
\label{eqMarkjumpdynfree8}
\int_{S\XC_{\ge 2}} [\zeta^* (\z) \nu (\z;d\y) -\zeta^* (\y) \tilde \nu (\z;d\y)]=0
\end{equation}
for all $\z \in \S\XC_{\ge 2}$, which is the equilibrium condition \eqref{eqMarkjumpdyncomp2equ}.

This condition \eqref{eqMarkjumpdynfree8} is called the {\it complex balance} condition
for evolution \eqref{eqMarkjumpdyncomp61}. As was shown, together with \eqref{eqMarkjumpdynfree4},
it is sufficient for evolution \eqref{eqMarkjumpdyncomp61} to 'respect' thermodynamics:  $\dot F(c(.))\le 0$.

In particular, for the pair-interaction dynamics  \eqref{eqMarkjumpdyncomp10pair} and \eqref{eqMarkjumpdyncomp12pair},
the complex balance condition takes the forms
\begin{equation}
\label{eqcompbalpair1}
\int_{S\XC^2}\left[\zeta^* (\y) \nu (\y;\x)-\zeta^* (\x) \nu (\x;\y)\right] d\y=0, \quad \x\in S\XC^2,
\end{equation}
and respectively
\begin{equation}
\label{eqcompbalpair2}
\int_{\Om}\left[\zeta^* (\y) B (\y;\om)-\zeta^* (\x) B (\x;\om)\right] d\om=0, \quad \x\in S\XC^2,
\end{equation}
and the evolution of the free energy
\begin{equation}
\begin{split}
\dot F(c(.))=&\int_X \frac{\de F}{\de c(x)}L^* c(x) M(dx)\\
\label{eqMarkjumpdynpairreac3} &+\frac14\int_{X^2}\int_{X^2}\left[\frac{\de F}{\de
c(y_1)}+\frac{\de F}{\de c(y_2)} -\frac{\de F}{\de c(x_1)}-\frac{\de F}{\de
c(x_2)}\right] \zeta (\x;c) \nu (\x;\y) M(d\x) M(d\y)
\end{split}
\end{equation}
and respectively
\begin{equation}
\begin{split}
\dot F(c(.))=&\int_X \frac{\de F}{\de c(x)}L^* c(x) M(dx)\\
\label{eqMarkjumpdynpair1reac3} &+\frac14\int_{X^2}\int_{X^2}\left[\frac{\de F}{\de
c(y_1)}+\frac{\de F}{\de c(y_2)} -\frac{\de F}{\de c(x_1)}-\frac{\de F}{\de
c(x_2)}\right] \zeta (\x;c) B(\x, \om) M(d\x) d\om,
\end{split}
\end{equation}
where
\[
\y=(y_1,y_2)=G^{\om} (\x).
\]

Of course \eqref{eqMarkjumpdynfree8} holds if
\begin{equation}
\label{eqMarkjumpdynfree9}
\zeta^* (\y) \nu (\y;\x)-\zeta^* (\x) \nu (\x;\y)=0
\end{equation}
for all $\x,\y \in S\XC_{\ge 2}$. This more restrictive condition is called the
 {\it detailed balance} condition for \eqref{eqMarkjumpdyncomp61}.

Turning to more general evolution \eqref{eqMarkjumpdyncomp6}, \eqref{eqMarkjumpdynfree2}
we shall reduce our attention only to MAL evolution,
where
\[
c(x)=c^*(x)\exp \left\{\frac{\de F}{\de c(x)}\right\}
\]
for the free energy in the perfect form \eqref{eqfreeengen11}. This
makes the notations for $\zeta(\x)$ consistent with
the notations $\zeta_1(x)=c(x)$. Consequently,
introducing $\theta$ by the equation
\begin{equation}
\begin{split}
\theta(\la)=&\int_{S\XC}\int_{S\XC} M(d\z) \zeta^* (\z) \nu (\z;d\y)
 \exp \left\{ \la \left(\frac{\de F}{\de c(\z)}\right)^{\oplus}
+(1-\la) \left(\frac{\de F}{\de c(\y)}\right)^{\oplus}\right\}\\
\label{eqdefGorbanThetaMal} =&\int_{S\XC}\int_{S\XC} M(d\z) \zeta^* (\z) \nu (\z;d\y)
\left[\frac{\c(\z)}{\c^*(\z)}\right]^{\la}\left[\frac{\c(\y)}{\c^*(\y)}\right]^{1-\la},
\end{split}
\end{equation}
yields again \eqref{eqMarkjumpdynfree3}. Moreover, condition $\theta(0)=\theta (1)$
becomes equivalent to condition \eqref{eqMarkjumpdyncomp1equ}, which is the {\it complex balance condition} for general MAL.

\subsection{Points of equilibrium}

Let us start with the MAL dynamics. As we know already, then $c^*$ is an equilibrium
point. Are there other (positive) equilibrium points? Assume the complex balance
condition \eqref{eqMarkjumpdyncomp1equ} holds, and let $c(x)$ be an equilibrium. Then we
have $\theta(0)=\theta(1)$ and $\theta'(1)=0$, which together with convexity of $\theta$
implies that $\theta(\la)$ is a constant (for given $c$). Hence $\theta''(\la) =0$, and
consequently, by \eqref{eqdefGorbanThetasecder},
\begin{equation}
\label{eqdefGorbanThetasecder1}
 \left(\frac{\de F}{\de c(\z)}\right)^{\oplus}- \left(\frac{\de F}{\de c(\y)}\right)^{\oplus}=0
\end{equation}
on the support of the measure $ M(d\z) \zeta^* (\z) \nu (\z;d\y)$, which coincides with the support of the measure
 $M(d\z) \nu (\z;d\y)$ if all $\zeta^*$ are strictly positive.

 In particular, it implies the following. Suppose the complex balance condition \eqref{eqMarkjumpdyncomp1equ} holds for
 a MAL evolution, all $\zeta^*$ are strictly positive, the evolution preserves the number of particles
 and the measure $\nu(\z,d\y)M_k(d\z)$ on $X^k\times X^k$ has the full support
 for at least one $k\ge 1$. Then $c$ is an equilibrium if and only if $\de F/\de c(x)$ is a constant (as a function of $x$)
 and $L^*c=0$, and hence, buy the structure of $F$, if and only if $c(x)$ coincides with $c^*$ up to a multiplicative constant.
Alternatively, assume $\zeta^*>0$, \eqref{eqMarkjumpdyncomp1equ} holds and the measure $\nu(\z,d\y)M_k(d\z)$ on $X^l\times X^k$
 has the full support for at least one pair $k\neq l$. Then $c^*$ is the only (positive) equilibrium.

Turning to evolution \eqref{eqMarkjumpdyncomp61} and \eqref{eqMarkjumpdynfree1} we can conclude similarly that
 if complex balance condition \eqref{eqMarkjumpdynfree8} holds,
 all $\zeta^*(\z)$ for $\z\in S\XC$ are strictly positive, the evolution preserves the number of particles
 and the measure $\nu(\z,d\y)M_k(d\z)$, $\z,y\in X^k$, has the full support
 for at least one $k\ge 1$, then $c$ is an equilibrium if and only if $\de F/\de c(x)$ is a constant (as a function of $x$)
 and $L^*c=0$. Alternatively, assume $\zeta^*>0$, \eqref{eqMarkjumpdynfree8} holds and the measure $\nu(\z,d\y)M_k(d\z)$ on $X^l\times X^k$
 has the full support for at least one pair $k\neq l$. Then $c$ is equilibrium if and only if  $\de F/\de c(x)=0$ and $L^*c=0$.

\subsection{Comments on the transformations of the free energy}

The GMAL evolution \eqref{eqMarkjumpdyncomp6} will not be changed if we make a linear shift of the free energy changing $F$ to
\[
\tilde F= F+\int \om (x) c(x) M(dx)
\]
with some $\om (x)$ and simultaneously change $\zeta^*$ to
\[
\tilde \zeta^*(\x)=\zeta(\x) \exp\left\{ -\sum_{l=1}^j  \om (x_l)\right\}, \quad \x=(x_1,
\cdots , x_j)\in X^j.
\]
Reducing our attention for simplicity to evolution \eqref{eqMarkjumpdyncomp61},  \eqref{eqcondpartnotcomp}, suppose the  complex
balance condition \eqref{eqMarkjumpdyncomp2equ} does not hold. The natural question arises whether
we can find a function $\om$ such that for new $\tilde F, \tilde \zeta^*$ it becomes valid, that is

\begin{equation}
\label{eqMarkjumpdyncomp2equlinch}
\int_{S\XC_{\ge 2}} [\zeta^*(\y)\exp\{-\om^{\oplus} (\y)\} \tilde \nu (\x, d\y)
 -\zeta^* (\x)\exp\{-\om^{\oplus} (\y)\} \nu (\x, d\y)]=0, \quad \x\in S\XC_{\ge 2},
\end{equation}
where
\[
\om ^{\oplus}(\x)=\om (x_1)+\cdots +\om (x_k), \quad \x=(x_1, \cdots, x_k).
\]

In discrete setting this question can be effectively answered algebraically by the
so-called deficiency zero theorem \cite{Feinberg1972,GorbanShahzad2011}. In our setting
let us note only that, if all initial $\zeta^*$ had a product form, this question reduces
to the directly verifiable question on whether the products
\[
\tilde \zeta^*(\x)=c^*(x_1) \cdots c^*(x_k)
\]
satisfy \eqref{eqMarkjumpdyncomp2equ}.

\section{More general free energy for compounds}

Motivated by Morimoto's Theorem \ref{MorimotoHtheoremext} \cite{Morimoto1963}, it is
natural to use for the free energy of the  compounds in \eqref{eqfreeengenmixcomp} the
thermodynamic Lyapunov function \eqref{eqdefgenentfu} generalizing
\eqref{eqfreeengenmixcomp} to
\begin{equation}
\label{eqfreeengenmixcomp1}
F_{tot}(c, \zeta_2, \zeta_3, \cdots)
=F(c (.)) +\sum_{j\ge 2} \int_{SX^j} H_h(\zeta_j (\x)\| \zeta_j^*(\x)) M_j(d\x).
\end{equation}

The condition of extremality \eqref{eqfreeengenmixcompextr} extends to
\begin{equation}
\label{eqfreeengenmixcompextrext}
-\int_X \frac{\de F}{\de c(x)}\int_{S\XC} \om (x, \y)\mu(x, d\y)M(dx)
+\int_{S\XC_{\ge 2}} \om (\x) h'\left( \frac{\zeta(\x)}{\zeta^*(\x)}\right) \M(d\x)=0
\end{equation}
and the equilibrium quantities \eqref{eqfreeengenmixcompextr3} become

\begin{equation}
\label{eqfreeengenmixcompextr3ext}
\zeta_j(\x;c)=\zeta^*_j(\x) g\left( \sum_{l=1}^j  \frac{\de F}{\de c(x_l)} \right),
\end{equation}
where $g$ is the inverse function to $h'$. Recall that $h$ was assumed convex on $\R^+$ and hence $h'$
is an increasing function $(0,\infty)\to (a,b)$ with some (finite or infinite) interval $(a,b)$.
Hence $g$ is an increasing function on $(a,b)$. Let
\[
G(x)=\int_{y_0}^x g(y) dy
\]
with some $y_0\in [a,b]$. Then $G$ is a concave function. In particular,
for the Burg relative entropy \eqref{eqdefgenentfuBurg}, $h(x)=-\ln x$ and $h'(x()=-1/x$ is self-inverse, so that $g(y)=-1/y$.
Formula \eqref{eqfreeengenmixcompextr3ext} become
\begin{equation}
\label{eqfreeengenmixcompextr3extBurg}
\zeta_j(\x;c)=-\frac{\zeta^*_j(\x)}{ \sum_{l=1}^j  c^*(x_l)/c(x_l)},
\end{equation}
and one can choose $G(x)=-\ln (x)$.

Dynamics equations \eqref{eqMarkjumpdyncomp6} or \eqref{eqMarkjumpdyncomp61} remain the same, though of course
with $\zeta$ of form \eqref{eqfreeengenmixcompextr3ext} rather than \eqref{eqfreeengenmixcompextr3}.
Consequently the evolution of the free energy remains the same, that is \eqref{eqMarkjumpdynfree1} or \eqref{eqMarkjumpdynfree2}.
The only thing needed a modification is the function $\theta$.
Let us restrict the discussion to evolution \eqref{eqMarkjumpdyncomp61} and \eqref{eqMarkjumpdynfree1} only
(that is, with restriction \eqref{eqcondpartnotcomp}) and define
\begin{equation}
\label{eqdefGorbanTheta1}
\theta(\la) =\int_{S\XC_{\ge 2}}\int_{S\XC_{\ge 2}} M(d\z) \zeta^* (\z) \nu (\z;d\y)
G \left( \la \left(\frac{\de F}{\de c(\z)}\right)^{\oplus}
+(1-\la) \left(\frac{\de F}{\de c(\y)}\right)^{\oplus}\right),
\end{equation}
so that
\begin{equation*}
\begin{split}
\theta '(\la)=\frac{d}{d\la} \theta (\la;c)= &\int_{S\XC_{\ge 2}}\int_{S\XC_{\ge 2}}
M(d\z) \zeta^* (\z) \nu (\z;d\y)
 \left[ \left(\frac{\de F}{\de c(\z)}\right)^{\oplus}- \left(\frac{\de F}{\de
 c(\y)}\right)^{\oplus}\right]\\
&\times g\left\{ \la  \left(\frac{\de F}{\de c(\z)}\right)^{\oplus} +(1-\la)
\left(\frac{\de F}{\de c(\y)}\right)^{\oplus} \right).
\end{split}
\end{equation*}
Then we get again $\theta''(\la)\ge 0$ and \eqref{eqMarkjumpdynfree3}.
Again condition $\theta (0)= \theta (1)$ turns out to be sufficient
for the decrease of the free energy by the evolution, and we finally conclude that if the
 linear combinations of the functions
 \[
G \left(\left(\frac{\de F}{\de c(\y)}\right)^{\oplus} \right)
\]
for all continuous functions $c$ are dense in the space of continuous functions on $S\XC_{\ge 2}$,
condition $\theta (0)= \theta (1)$ is equivalent to \eqref{eqMarkjumpdynfree8}, i.e. to the complex balance condition.

\section{Diffusion approximation}

\subsection{Binary mechanisms of diffusion}

Let us consider a lattice $ h\Z^d$, $h>0$, in $\R^n$ equipped with the standard basis $e_1, \cdots, e_n$.
 To each cell or site $x=h(j_1, \cdots , j_n)$ there is attached a locally compact state space $V$
 specifying the possible types of particles. Fixing some measure $M(dv)$ in $V$ we can speak about the concentration
 $c(x,v)$ of particles of type $v$ at the site $x$. The concentration of pairs will be considered
 with respect to the product measure on $V^2$.

 By $N(h,x)$ let us denote the set of neighboring cells to $x$, that is
 \[
 N(x,h)=\{y=x\pm h e_i, \quad i=1, \cdots , n\}.
 \]

 We start here with modeling only the movement
of the particles around $h\Z^d$, when no change of type is possible.
  We shall assume that only particles in neighboring cells
can interact and that the interaction is pairwise (which is mostly observed in practice).
We shall also assume that our lattice is homogenous in the sense that all rate constants, equilibria concentrations, etc,
do not depend on the site.

 There are three natural mechanisms of transitions between any
chosen pair of neighboring cells $(x,y=x+he_i)$, which we shall also denoted $I,II$
\cite{GorbanSarkisyan2011}:

Exchange: $(v^I, w^{II}) \to (v^{II}, w^I)$, that is, particles of type $v,w$ exchange places;

Clustering: $(v^I, w^{II}) \to (v^{II}, w^{II})$, that is, a particle from one cell attracts a particle from another one;

Repulsion: $(v^I, w^I) \to (v^I, w^{II})$, which is the inverse process to clustering.

In the spirit of our general approach, we shall assume that any pair of particles, before an interaction, should form a compound
 of two particles. Moreover, the interaction between compounds is linear
 and the number of pairs are in fast equilibrium with the concentration
of free particles according to the rule  \eqref{eqfreeengenmixcompextr4}, that is, the concentration $\zeta((x,v),(y,w))$
of pairs in two neighboring cells $(x,y=x+he_i)$ equals
\begin{equation}
\label{eqconseqdif} \zeta((x,v),(y,w);c)=\zeta^*(v,w) \exp \left\{ \frac{\de F}{\de
c(x,v)}+ \frac{\de F}{\de c(y,w)} \right\},
\end{equation}
with some equilibrium $\zeta^*$ (not depending on $x,y$ by the assumed homogeneity) and a free energy functional $F$.

For the case of perfect free energy
\begin{equation}
\label{eqfreeperflat}
F(c(.))=\sum_{x\in h\Z^n} \int_V c(x,v) \left(\ln \frac{c(x,v)}{c^*(v)}-1\right) M(dv),
\end{equation}
this turns to the MAL dependence
\begin{equation}
\label{eqconseqdifMAL}
\zeta((x,v),(y,w);c)=\frac{\zeta^*(v,w)}{c^*(v)c^*(w)} c(x,v) c(y,w).
\end{equation}

Furthermore, as we assumed $M_2$ to have a product form, one can take the transition kernels $\mu$ from
\eqref{eqMarkjumpdyncomp8pair}, that is $\mu (x, dy) =M (dy)$ and hence \eqref{eqasscompbasmaes3} becomes
\begin{equation}
\begin{split}
c(x,v) =&\zeta (x,v) + \sum_{y\in N(h,x)} \int_V \zeta ((x,v),(y,w)) M(dw)\\
\label{eqconseqdifinpairs} =&\zeta (x,v)+\sum_{i=1}^n \int_V [\zeta
((x,v),(x+he_i,w))+\zeta ((x,v),(x-he_i,w))] M(dw).
\end{split}
\end{equation}

\subsection{Exchange}

Let us start with the reaction of exchange. The linear reaction of the concentrations $\zeta((x,v);(y,w))$
due to the exchange mechanism between the cells $(x,y)$ is described by the equation
\[
\dot \zeta((x,v),(y,w))=k(v,w)[\zeta((y,v),(x,w))-\zeta((x,v),(y,w))].
\]
Here the rates $k(v,w)$ do not depend on the sites by homogeneity, but it can depend on the order of the arguments $v,w$.
The r.h.s. of this equation describes the flux of particles along the edge $(x,y)$, or, having in mind another equivalent visual picture,
through the border of the cells centered at $x$ and $y$.

Assuming only the exchange mechanism in the system and the MAL condition \eqref{eqconseqdifMAL} it follows that
\begin{equation*}
\begin{split}
\dot c(x,v)=&\sum_{i=1}^n \int_V [\dot \zeta ((x,v),(x+he_i,w))+\dot \zeta
((x,v),(x-he_i,w))] M(dw)\\
=&\sum_{i=1}^n \int_V k(v,w)\frac{\zeta^*(v,w)}{c^*(v)c^*(w)}M\\
&\times [c(x+he_i,v)c(x,w)-c(x,v)c(x+he_i,w)+c(x-he_i,v)c(x,w)-c(x,v)c(x-he_i,w)].
\end{split}
\end{equation*}
Introducing the normalized rates
\[
\phi (v,w)=  k(v,w)\frac{\zeta^*(v,w)}{c^*(v)c^*(w)}
\]
and expanding the functions $c$ in Taylor series up to the second order we obtain
in the first nontrivial approximation
\begin{equation}
\label{eqecolairexcgh}
\dot c(x,v) =h^2 \int \phi (v,w) [\De c(x,v)c(x,w)-\De c(x,w) c(x,v)] M(dw),
\end{equation}
where the Laplacian $\De$ acts on the first variable of $c(x,v)$.
Allowing additionally the evolution of free particles according to the simplest linear dynamics
\[
\dot c(x,v)=\sum_{y\in N(h,x)} \int_V k(v) [c(y,v)-c(x,v)] M(dw)
\]
yields in the first approximation the dynamics
\begin{equation}
\label{eqecolairexcgh1}
\dot c(x,v) =h^2 k(v) \De c(v,x) +h^2 \int \phi (v,w) [\De c(x,v)c(x,w)-\De c(x,w) c(x,v)] M(dw),
\end{equation}
which can be equivalently written in the form
\begin{equation}
\label{eqecolairexcgh2}
\dot c(x,v) =h^2 k(v) {\rm {div}} \nabla c(x,v) +h^2 \int \phi (v,w) {\rm {div}} [\nabla c(x,v)c(x,w)-\nabla c(x,w) c(x,v)] M(dw),
\end{equation}
(with derivations acting on the first variable of $c$).

To get a proper limiting equation one has to assume, of course, that $k$ and $\phi$ scale appropriately with $h$, so that
the limits
\[
K(v)=\lim_{h\to 0} h^2 k(v), \quad  \Phi(v,w)=\lim_{h\to 0} h^2 \phi(v,w)
\]
exist, in which case the limiting equation takes the form
\begin{equation}
\label{eqecolairexcgh20}
\dot c(x,v) =K(v) {\rm {div}} \nabla c(x,v) +\int \Phi (v,w) {\rm {div}} [\nabla c(x,v)c(x,w)-\nabla c(x,w) c(x,v)] M(dw).
\end{equation}
 The same remark concerns all limiting equations below.

For a more general free energy $F(c(.))$ the evolution becomes
\begin{equation*}
\begin{split}
 \dot c(x,v) =&\sum_{i=1}^n \int_V k(v,w)\zeta^*(v,w)M(dw)\\
 &\times \left[\exp\left\{ \frac{\de F}{\de c(x+he_i,v)}+\frac{\de F}{\de c(x,w)}\right\} -
\exp\left\{ \frac{\de F}{\de c(x+he_i,w)}+\frac{\de F}{\de c(x,v)}\right\} \right.\\
&+\left. \exp\left\{ \frac{\de F}{\de c(x-he_i,v)}+\frac{\de F}{\de c(x,w)}\right\} -
\exp\left\{ \frac{\de F}{\de c(x-he_i,w)}+\frac{\de F}{\de c(x,v)}\right\}\right].
\end{split}
\end{equation*}
Expanding the variational derivatives in Taylor series up to the second order yields
\begin{equation*}
\begin{split}
\exp&\left\{ \frac{\de F}{\de c(x+he_i,v)}+\frac{\de F}{\de c(x,w)}\right\}\\
&=\exp\left\{ \frac{\de F}{\de c(x,v)}+\frac{\de F}{\de c(x,w)}\right\}
\exp\left\{h\frac{\pa}{\pa x_i}\frac{\de F}{\de c(x,v)}+\frac12h^2\frac{\pa^2}{\pa
x_i^2}\frac{\de F}{\de c(x,v)}\right\}\\
&=\exp\left\{ \frac{\de F}{\de c(x,v)}+\frac{\de F}{\de c(x,w)}\right\} \left[1
+h\frac{\pa}{\pa x_i}\frac{\de F}{\de c(x,v)}+\frac12h^2\frac{\pa^2}{\pa x_i^2}\frac{\de
F}{\de c(x,v)} +\frac12h^2 \left(\frac{\pa}{\pa x_i}\frac{\de F}{\de c(x,v)}\right)^2
\right]
\end{split}
\end{equation*}
and similar with other terms.
Thus one sees that zero-order and first order terms again cancel, and the second order terms
yield the equation

\begin{equation}
\begin{split}
\dot c(x,v) =&h^2 \int k(v,w)\zeta^*(v,w) \exp\left\{ \frac{\de F}{\de c(x,v)}+\frac{\de
F}{\de c(x,w)}\right\} M(dw)\\
\label{eqecolairexcgh3} &\times\left(\De \frac{\de F}{\de c(x,v)} -\De \frac{\de F}{\de
c(x,w)} +\left|\nabla \frac{\de F}{\de c(x,v)}\right|^2-\left|\nabla \frac{\de F}{\de
c(x,w)} \right|^2\right),
\end{split}
\end{equation}
which can also be written in the divergence form
\begin{equation}
\begin{split}
\dot c(x,v) =& {\rm {div}} \int \phi(v,w) M(dw) \\
\label{eqecolairexcgh4} &\times \left[ \exp\left\{ \frac{\de F}{\de c(x,w)}\right\}
\nabla \exp\left\{ \frac{\de F}{\de c(x,v)}\right\} - \exp\left\{ \frac{\de F}{\de
c(x,v)}\right\} \nabla \exp\left\{ \frac{\de F}{\de c(x,w)}\right\}\right].
\end{split}
\end{equation}
with
\[
\phi(v,w)= h^2 k(v,w)\zeta^*(v,w).
\]

Let us calculate the evolution of the thermodynamic Lyapunov function $F(c(.))$
along the evolution \eqref{eqecolairexcgh3}. We shall consider the unbounded lattice $h\Z^d$ and its limit $\R^d$
(alternatively, one can work with finite volume assuming appropriate boundary conditions, say periodic).  We have
\[
\dot F(c(.)) =\int_{\R^d}\int_V  \frac{\de F}{\de c(x,v)} \dot c(x,v) \, dx M(dv).
\]
Substituting \eqref{eqecolairexcgh4} and using the symmetry with respect to the integration variable $v,w$ we get
\begin{equation*}
\begin{split}
\dot F(c(.)) =&\frac12 \int_{\R^d}\int_{V^2}  \left[\frac{\de F}{\de
c(x,v)}\phi(v,w)-\frac{\de F}{\de c(x,w)} \phi (w,v)\right] dx M(dv)M(dw)\\
 &\times {\rm {div}} \left[ \exp\left\{ \frac{\de F}{\de c(x,w)}\right\} \nabla \exp\left\{ \frac{\de F}{\de c(x,v)}\right\}
- \exp\left\{ \frac{\de F}{\de c(x,v)}\right\} \nabla \exp\left\{ \frac{\de F}{\de c(x,w)}\right\}\right],
\end{split}
\end{equation*}
or, integrating by parts in $x$,
\begin{equation}
\begin{split}
\dot F(c(.)) =&-\frac12 \int_{\R^d}\int_{V^2} \nabla  \left[\frac{\de F}{\de
c(x,v)}\phi(v,w)-\frac{\de F}{\de c(x,w)} \phi (w,v)\right] dx M(dv)M(dw)\\
 \label{eqecolairexcgh5}
 &\times \nabla \left[\frac{\de F}{\de c(x,v)}-\frac{\de F}{\de c(x,w)} \right]
\exp\left\{ \frac{\de F}{\de c(x,v)} + \frac{\de F}{\de c(x,w)}\right\}.
\end{split}
\end{equation}
Hence, if the {\it detailed balance condition}
\[
\phi(v,w)=\phi(w,v)
\]
 holds (or equivalently $k(v,w)=k(w,v)$), then
\begin{equation}
\begin{split}
\dot F(c(.)) =&-\frac12 \int_{\R^d\times V^2} \phi (v,w) \left|\nabla  \left(\frac{\de
F}{\de c(x,v)}-\frac{\de F}{\de c(x,w)} \right)\right|^2 \, dx M(dv)M(dw)\\
 \label{eqecolairexcgh6} &\times \exp\left\{ \frac{\de F}{\de c(x,v)} + \frac{\de F}{\de
c(x,w)}\right\},
\end{split}
\end{equation}
which is clearly non-positive.

\subsection{Repulsion and attraction}

Let us turn to attraction - repulsion interactions. Introducing the rate constants $k_{atr}(v,w)$, describing the process that pushes
a particle $v$ to a neighboring particle $w$, and $k_{rep}$, describing the process with a which a particle $w$ can kick out a particle $v$
(siting at the same site as $w$) to a neighboring site,
 we can write the following linear evolution of the concentrations $\zeta((x,v);(y,w))$
due to attraction -repulsion mechanism between the cells $(x,y)$:
\begin{equation*}
\begin{split}
\dot \zeta((x,v),(y,w))=&k_{rep}(w,v)\zeta((x,v),(x,w)) +k_{rep}(v,w)\zeta((y,v),(y,w))\\
&-(k_{atr}(v,w)+k_{atr}(w,v))\zeta((x,v),(y,w)),
\end{split}
\end{equation*}
\begin{equation*}
\begin{split}
\dot \zeta((x,v),(x,w))=&k_{atr} (w,v)\zeta((x,v),(y,w))+k_{atr}(v,w)\zeta((y,v),(x,w))\\
&-(k_{rep}(v,w)+k_{rep}(w,v))\zeta((x,v),(x,w)).
\end{split}
\end{equation*}
It is worth noting that $k_{rep}(v,w)$ and $k_{atr}(v,w)$ need not be symmetric functions of $v,w$. Even more so,
there are natural situations with, say,  $k_{atr}(v,w)>0$ and $k_{atr}(w,v)=0$, which means that $v$ is a mobile particle
and $w$ is not.

As now we shall have to take into accounts the compounds of particles sitting on the same site,
\eqref{eqconseqdifinpairs} generalizes to
\begin{equation}
\label{eqconseqdifinpairsgen}
c(x,v)=\zeta (x,v)+\int_V [\zeta ((x,v),(x,w))+\sum_{i=1}^n  \zeta ((x,v),(x+he_i,w))+\zeta ((x,v),(x-he_i,w))] M(dw)
\end{equation}
Moreover, fast equilibrium condition \eqref{eqconseqdif} should be supplemented by the condition
\begin{equation}
\label{eqconseqdifsamesit} \zeta((x,v),(x,w);c)=\tilde \zeta^*(v,w) \exp \left\{
\frac{\de F}{\de c(x,v)}+ \frac{\de F}{\de c(x,w)} \right\},
\end{equation}
with some $\tilde \zeta^*(v,w)$ that can be different from $\zeta^*(v,w)$,
which in the case of the perfect free energy
turns to the MAL dependence
\begin{equation}
\label{eqconseqdifMALsit}
\zeta((x,v),(x,w);c)=\frac{\tilde \zeta^*(v,w)}{c^*(v)c^*(w)} c(x,v) c(x,w).
\end{equation}

Thus taking into account only the attraction-repulsion mechanism, using again
for simplicity the MAL condition \eqref{eqconseqdifMAL},
and introducing the normalized rates
\[
\phi_{atr} (v,w)=  k_{atr}(v,w)\frac{\zeta^*(v,w)}{c^*(v)c^*(w)},
\quad \phi_{rep} (v,w)=  k_{rep}(v,w)\frac{\tilde \zeta^*(v,w)}{c^*(v)c^*(w)},
\]
the evolution of the concentrations becomes

\begin{equation*}
\begin{split}
\dot c(x,v)=&\sum_{i=1}^n \int_V [\phi_{atr}(w,v) c(x,v)c(x+he_i,w)+\phi_{atr}(v,w)
c(x+he_i,v)c(x,w)\\
 &-(\phi_{rep}(w,v)+\phi_{rep}(v,w))c(x,v)c(x,w)]M(dw)\\
 &+\sum_{i=1}^n \int_V [\phi_{atr}(w,v) c(x,v)c(x-he_i,w)+\phi_{atr}(v,w)
 c(x-he_i,v)c(x,w)\\
 &-(\phi_{rep}(w,v)+\phi_{rep}(v,w))c(x,v)c(x,w)]M(dw)\\
 &+\sum_{i=1}^n \int_V [\phi_{rep}(w,v) c(x,v)c(x,w)+\phi_{rep}(v,w)
 c(x+he_i,v)c(x+he_i,w)\\
 &-(\phi_{atr}(w,v)+\phi_{atr}(v,w))c(x,v)c(x+he_i,w)]M(dw)\\
 &+\sum_{i=1}^n \int_V [\phi_{rep}(w,v) c(x,v)c(x,w)+\phi_{rep}(v,w)
 c(x-he_i,v)c(x-he_i,w)\\
 &-(\phi_{atr}(w,v)+\phi_{atr}(v,w))c(x,v)c(x-he_i,w)]M(dw).
\end{split}
\end{equation*}

Expanding the functions $c$ in Taylor series, we see that the terms of zero-order and
first-order in $h$ cancel. Expanding up to the second order we obtain the equation
\begin{equation*}
\begin{split}
\dot c(x,v)=&h^2 \sum_{i=1}^n \int_V \left[\phi_{atr}(w,v) c(x,v) \frac{\pa ^2c}{\pa
x_i^2} (x,w) +\phi_{atr}(v,w) c(x,w) \frac{\pa ^2c}{\pa x_i^2} (x,v)\right]M(dw)\\
&+h^2\sum_{i=1}^n \int_V \phi_{rep}(v,w) \left[ \frac{\pa ^2c}{\pa x_i^2} (x,c) c(x,w)
+\frac{\pa ^2c}{\pa x_i^2} (x,w)c(x,v) +2\frac{\pa c}{\pa x_i} (x,v)\frac{\pa c}{\pa x_i}
(x,w)\right]M(dw),
\end{split}
\end{equation*}
or in concise notations
\begin{equation}
\begin{split}
\dot c(x,v) =&h^2 \int [\phi_{atr}(v,w)c(x,w)\De c(x,v) +\phi_{atr}(w,v)c(x,v)\De
c(x,w)]M(dw)\\
\label{eqdifapprrepatr} &+h^2 \int \phi_{rep}(v,w)\De [c(x,w)c(x,v)]
M(dw).
\end{split}
\end{equation}
The 'repulsion' part (with vanishing $\phi_{atr}$) of this equation can also be written in the divergence form:
\begin{equation}
\label{eqdifapprrepatr1}
\dot c(x,v) =h^2 {\rm {div}} \int \phi_{rep}(v,w)\nabla [c(x,v)c(x,w)] M(dw).
\end{equation}

Generalizing, as above for the exchange mechanism, to more general free energy $F(c(.))$,
equation \eqref{eqdifapprrepatr} generalizes to

\begin{equation}
\begin{split}
\dot c(x,v) = &\int
 \left[\phi_{atr}(v,w)\De \frac{\de F}{\de c(x,v)} +\phi_{atr}(w,v)\De \frac{\de F}{\de c(x,w)}\right]
 \exp\left\{ \frac{\de F}{\de c(x,v)}+\frac{\de F}{\de c(x,w)}\right\} M(dw)\\
\label{eqdifapprrepatr2}& + \int \phi_{rep}(v,w) \De \exp\left\{ \frac{\de F}{\de
c(x,v)}+\frac{\de F}{\de c(x,w)}\right\} M(dw).
\end{split}
\end{equation}
where
\[
\phi_{atr} (v,w)= h^2 k_{atr}(v,w)\zeta^*(v,w),
\quad \phi_{rep} (v,w)= h^2 k_{rep}(v,w)\tilde \zeta^*(v,w).
\]

Similarly to the calculations with exchange mechanism above, we find the following law of
the evolution of $F$ due to the repulsion mechanism \eqref{eqdifapprrepatr2} (taking vanishing $\phi_{atr}$ in \eqref{eqdifapprrepatr2}):

\begin{equation}
\begin{split}
\dot F(c(.)) =&-\frac12 \int_{\R^d}\int_{V^2} \nabla \left[\frac{\de F}{\de
c(x,v)}\phi_{rep}(v,w)+\frac{\de F}{\de c(x,w)} \phi_{rep} (w,v)\right] \, dx
M(dv)M(dw)\\
\label{eqrepmechLyap1}
 &\times \nabla \left[\frac{\de F}{\de c(x,v)}+\frac{\de F}{\de c(x,w)} \right]
\exp\left\{ \frac{\de F}{\de c(x,v)} + \frac{\de F}{\de c(x,w)}\right\}.
\end{split}
\end{equation}
Hence, if the {\it detailed balance condition}
\[
\phi_{rep}(v,w)=\phi_{rep}(w,v)
\]
 holds (or equivalently $k_{rep}(v,w)=k_{rep}(w,v)$), then
\begin{equation}
\begin{split}
 \dot F(c(.)) =&-\frac12 \int_{\R^d\times V^2} \phi_{rep} (v,w)
 \left|\nabla  \left(\frac{\de F}{\de c(x,v)}+\frac{\de F}{\de c(x,w)} \right)\right|^2 \, dx
 M(dw)M(dv)\\
\label{eqecolairexcgh6gen} &\times \exp\left\{ \frac{\de F}{\de c(x,v)} + \frac{\de
F}{\de c(x,w)}\right\} ,
\end{split}
\end{equation}
which is clearly non-positive.

\subsection{Diffusion combined with other reactions}

Suppose that on the sites of the lattice the particles can react according to \eqref{eqMarkjumpdyncomp61},
though only pairs of particles can interact producing only two or three particles.
Suppose also the free energy is perfect leading to MAL
with all equilibrium concentration normalized to unity and that the simplest
 product measure $M(dv)M(dw)$ on $V^2$ can be used to measure the concentration of pairs.
Then the total dynamics comprising diffusion along the spatial variable
(including one-particle diffusion, exchange and repulsion-attraction mechanism) and reactions on the sites becomes
 \begin{equation}\label{eqMarkgendifreac}
 \begin{split}
 \dot c(x&,v)= K(v) {\rm {div}} \nabla c(x,v)+L^* c(x,v)\\
  &+ {\rm {div}} \int_V \Phi (v,w)  [\nabla c(x,v)c(x,w)-\nabla c(x,w) c(x,v)] M(dw)\\
 &+{\rm {div}} \int_V \Phi_{rep}(v,w) [\nabla c(x,v)c(x,w)+\nabla c(x,w) c(x,v)] M(dw)\\
 &+\int [\Phi_{atr}(v,w)c(x,w)\De c(x,v) +\Phi_{atr}(w,v)c(x,v)\De c(x,w)]M(dw)\\
 &+\int_V M (dw) \sum_{k=2}^{\infty}
\int_{SV^k}\left[\prod_{j=1}^k c(x,u_j) \tilde \nu (v,w, du_1 \cdots u_k)
 -c(x,v)c(x,w) \nu (v, w, du_1 \cdots du_k)\right],
 \end{split}
\end{equation}
where all differentiations act on the $x$ variable and $L^*$ acts on the second variable.

Let us stress that the mathematical difficulties in rigorous study of this type of equations in general are enormous.
In particular, this type includes the full classical Boltzmann equation, for which the well-posedness is a well known open
problem.

As a simple interesting example let us describe the case of only two types of particles, $V=\{A,B\}$,
such that the particles of the second type $B$ are immobile (in particular, there is no exchange)
 and act only as catalysis for the branching of $A$.
If the death rate of $A$ is $\Phi_d$, the corresponding evolution of the concentration of
$A$ (the concentration of $B$ does not evolve in time) becomes
\begin{equation}
\begin{split}
 \dot c_A(x)= &K \De c_A(x)-\Phi_d c_A(x)
  + \Phi_{rep} {\rm {div}} [c_B(x)\De c_A(x)+c_A(x) \De c_B(x)] \\
\label{eqMarkgendifreac1}
 &+\Phi_{atr}c_B(x)\De c_A(x) +
 \sum_{k=2}^{\infty} c_B(x)c_A(x)(c^{k-1}(x)-1) \nu_k.
 \end{split}
\end{equation}

Equations of that type are actively studied now in econophysics as models for economic and biological growth,
the solutions having quite peculiar properties, see
e.g. \cite{YNRS}.

\section{Conclusion}

We studied the Michaelis--Menten--Stueckelberg limit (Figure~\ref{MMSlimit}) and found
the general form of a nonlinear evolutions describing transformations of particles in
this limit  which combines QSS and QE assumptions about transformations of intermediates.

The resulting evolution can be considered as a far reaching extension to arbitrary state
spaces of the theory  developed by  Michaelis and Menten for the simple enzyme kinetic
and by Stueckelberg for Boltzmann's gas with collisions.  It is developed both for pure
jump underlying processes and for their diffusive limits.  It is shown that the
corresponding (generalized) free energy monotonically decreases  whenever the evolution
satisfies either the detailed balance condition  or more generally a complex balance (or
cyclic balance) condition. The complex balance conditions follows from the Markov
microkinetics in the Michaelis--Menten--Stueckelberg limit.

\section{Appendix}

\subsection{On pure-jump Markov processes}

Let $X$ be a locally compact metric space.
A generator of an arbitrary pure-jump Markov process (Markov chain) on $X$ has the form
\begin{equation}
\label{eqMarkjumpgen1}
Lf (x)=\int (f(y)-f(x)) \nu (x,dy)
\end{equation}
with a stochastic kernel $\nu$. The dual operator on measures is
\begin{equation}
\label{eqMarkjumpgen2}
L^*\mu (dx)=\int [\nu (y,dx) \mu (dy) -\nu (x,dy) \mu (dx)],
\end{equation}
so that the evolution of the distributions of the Markov process specified by $L$ is
\[
\dot \mu =L^* \mu.
\]
Let a Radon measure $M(dx)$ (i.e. a Borel measure with all compact sets having a finite measure) be chosen on $X$.
we say that a bounded measure $\mu$ has the concentration or the density-function $c\in L^1(M)$ if
$\mu$ is absolutely continuous with respect to $M$ with the Radon-Nikodyme derivative being $c$, that is
\[
\int_V \mu(dx) =\int_V c(x) M(dx)
\]
for any Borel set $V$.
In order to be able to restrict the evolution $\dot \mu =L^* \mu$ on measures with the densities,
we have to make the following assumption:

The projection of the measure $\nu (y,dx)M(dy)$ on $x$, that is the measure $\int_{y\in X}\nu (y,dx)M(dy)$ on $X$,
is absolutely continuous with respect
to $M$ or equivalently (by the disintegration of measure theory)
 there exists a stochastic kernel $\tilde \nu (x,dy)$ such that
\begin{equation}
\label{eqdisintmes1}
\nu (y,dx)M(dy)=\tilde \nu (x, dy)M(dx).
\end{equation}

If this is the case,
\begin{equation}
\label{eqMarkjumpgen3}
L^*[c(x) M (dx)]=\int [c(y)\tilde \nu (x, dy) M (dx) -\nu (x,dy) c(x) M (dx)],
\end{equation}
and the evolution equation $\dot \mu =L^* \mu$ in terms of the concentrations becomes
\begin{equation}
\label{eqMarkjumpgen4}
\dot c(x)=\int_{y\in X} [c(y)\tilde \nu (x, dy) -c(x) \nu (x,dy)].
\end{equation}

\begin{remark}
The dual to \eqref{eqMarkjumpgen3} equation on functions is
\begin{equation}
\label{eqMarkjumpgen5}
\dot f(x)=\int_{y\in X} [f(y) -f(x)] \nu (x,dy).
\end{equation}
Its well-posedness (implying the well-posedness for \eqref{eqMarkjumpgen4}) is investigated
under rather general conditions on possibly unbounded $\nu$ in
\cite{Ko04}, \cite{Ko06a} and \cite{Ko10}.
\end{remark}

More generally, if we have $n$ locally compact metric spaces $X_j$, $j=1, \cdots, k$,
a generator of an arbitrary pure-jump Markov process on the disjoint union of these $X_j$ has the form
\begin{equation}
\label{eqMarkjumpgenmul1}
(Lf)_j (x_j)=\sum_{l=1}^k \int (f_l(y_l)-f_j(x_j)) \nu_{j\to l} (x_j,dy_l)
\end{equation}
with some stochastic kernels $\nu_{j\to l}$. The dual operator on measures becomes
\begin{equation}
\label{eqMarkjumpgenmul2}
(L^*\mu)_j (dx_j)=\sum_{l=1}^k \int [\nu_{l\to j} (y_l,dx_j) \mu (dy_l) -\nu_{j\to l} (x_j,dy_l) \mu (dx_j)].
\end{equation}
Extending \eqref{eqdisintmes1} we assume that
\begin{equation}
\label{eqdisintmesmul1}
\nu_{l\to j} (y_l,dx_j)M(dy_l)=\tilde \nu_{j\to l} (x_j, dy_l)M(dx_j).
\end{equation}

In this case the evolution of the distributions  $\dot \mu =L^* \mu$ with $\mu =(\mu_1, \cdots , \mu _k)$
can be restricted to the concentrations yielding the evolution
\begin{equation}
\label{eqMarkjumpgenmul3}
\dot c_j(x_j)=\sum_{l=1}^k \int_{X_l} [c_l(y_l)\tilde \nu_{j\to l} (x_j, dy_l) -c_j(x_j)\nu_{j\to l} (x_j,dy_l)].
\end{equation}
In the simplest case when all $\nu_{j\to l}(x_j, dy_l)$ have densities $\nu_{j\to l}(x_j, y_l)$ with respect to $M_l$,
\eqref{eqMarkjumpgenmul3} turns to
\begin{equation}
\label{eqMarkjumpgenmul4}
\dot c_j(x_j)=\sum_{l=1}^k \int_{X_l} [c_l(y_l)\nu_{l\to j} (y_l, x_j) - c_j(x_j) \nu_{j\to l} (x_j,y_l)] M(dy_l).
\end{equation}

Of course evolution \eqref{eqMarkjumpgenmul3} can be considered as a particular case of \eqref{eqMarkjumpgen4} if
$X$ is taken to be the disjoint union of spaces $X_j$.

Recall now that the concentration $c^*(x)$ is called an equilibrium for system \eqref{eqMarkjumpgen4}, if
\begin{equation}
\label{eqMarkjumpgeneq1}
\int_{y\in X} [c^*(y)\tilde \nu (x, dy) -c^*(x) \nu (x,dy)]=0.
\end{equation}
If this is the case, and assuming $c^*(x)>0$ everywhere, equation \eqref{eqMarkjumpgen4} rewrites equivalently as
\begin{equation}
\label{eqMarkjumpgen3witheq}
\dot c(x)=\int_{y\in X} c^*(y)\left[\frac{c(y)}{c^*(y)}-\frac{c(x)}{c^*(x)}\right] \tilde \nu (x, dy).
\end{equation}

For a convex smooth function $h(x)$ let us introduce the 'generalized entropy' function
\begin{equation}
\label{eqdefgenent}
H_h(c \| c^*)=\int  c^*(x) h\left( \frac{c (x)}{c^*(x)}\right) M(dx).
\end{equation}
Assuming that $c$ evolves according to  \eqref{eqMarkjumpgen3witheq} and that all integrals below are
 well defined, it follows that
\begin{equation}
\label{eqdefgenent1}
\frac{d}{dt}H_h(c \| c^*)=\int \int  h'\left( \frac{c (x)}{c^*(x)}\right)
 c^*(y)\left[\frac{c(y)}{c^*(y)}-\frac{c(x)}{c^*(x)}\right] \tilde \nu (x, dy) M(dx).
\end{equation}

Generalizing the concepts from the theory of Markov chains let us introduce the graph $(X,E)$
associated with evolution \eqref{eqMarkjumpgen3} such that the set of vertices $X$ coincides
with the state space $X$ and the edge $(x \to y)$ exists if the point $y$ belongs to the support
of the measure $\nu(x, dy)$. As usual, we say that the finite sequence $(y_0, y_1, \cdots , y_k)$ is a path
in this graph joining $y_0$ and $y_k$ if the edges $(y_{j-1} \to y_j)$ exist for all $j=1, \cdots, k$;
and that the graph is strongly connected if
for any pair of points $(y_0,y)$ there exist paths joining $y_0$ and $y$.

The following result is the extension of the
Morimoto H-theorem of finite state-space Markov chains to the continuous state-space:
\begin{proposition} \label{MorimotoHtheoremext} Under evolution \eqref{eqMarkjumpgen3witheq}, and assuming $c^*(x)>0$ everywhere,
\begin{equation}
\label{eq1MorimotoHtheorem}
\frac{d H_h(c \|c^*)}{dt} \le 0.
\end{equation}
Moreover, if the measure $M(dx)$ has the full support and the graph $(X,E)$ introduced above
is strongly connected, then the equality in \eqref{eq1MorimotoHtheorem} holds if and only if
the ratio $c(x)/c^*(x)$ is a constant.
\end{proposition}

\begin{proof}
As it follows from \eqref{eqMarkjumpgen2}, $\int L^*\mu (dx)=0$ for all $\mu$. In terms of equation
\eqref{eqMarkjumpgen3witheq} this rewrites as
\[
0=\int \dot c (x) M(dx) =\int_X c^*(y)\tilde \nu (x, dy) \left[\frac{c(y)}{c^*(y)}-\frac{c(x)}{c^*(x)}\right] M(dx)
\]
\begin{equation}
\label{eq2MorimotoHtheorem}
=\int_X c^*(y)\nu (y, dx) \left[\frac{c(y)}{c^*(y)}-\frac{c(x)}{c^*(x)}\right] M(dy)
\end{equation}
for any $c(x)$. Consequently, for any function $f$ (such that the integral below is well defined),
\begin{equation}
\label{eq3MorimotoHtheorem}
\int_X c^*(y) \nu (y, dx) \left[f(y)-f(x)\right] M(dy)=0.
\end{equation}
This identity allows one to rewrite \eqref{eqdefgenent1} as
\[
\frac{d}{dt}H_h(c \| c^*)=\int \int c^*(y) \nu (y, dx) M(dy)
 \]
\begin{equation}
\label{eqdefgenent2}
 \times
 \left[ h \left(\frac{c(x)}{c^*(x)}\right)-h\left(\frac{c(y)}{c^*(y)}\right)
 +h'\left( \frac{c (x)}{c^*(x)}\right)\left(\frac{c(y)}{c^*(y)}-\frac{c(x)}{c^*(x)}\right)\right],
\end{equation}
implying \eqref{eq1MorimotoHtheorem} by the convexity of $h$.

Finally, assuming $M$ has full support, it follow that  the equality in \eqref{eq1MorimotoHtheorem} holds if and only if
\[
\int \nu (y, dx)
 \left[ h \left(\frac{c(x)}{c^*(x)}\right)-h\left(\frac{c(y)}{c^*(y)}\right)
 +h'\left( \frac{c (x)}{c^*(x)}\right)\left(\frac{c(y)}{c^*(y)}-\frac{c(x)}{c^*(x)}\right)\right]=0
\]
for all $y$. Hence by convexity, $\frac{c(x)}{c^*(x)}=\frac{c(y)}{c^*(y)}$ for all $x$
from the support of $\nu(y,.)$.
 The final conclusion follows from the assumed connectivity of $(X,E)$.
\end{proof}

\subsection{Linking the concentration of particles and of compounds}

\begin{proposition}
\label{kernstoich}
The kernels $\mu$ in \eqref{eqasscompbasmaes} can be chosen in such a way that
if $\zeta_k (x_1, \cdots , x_k)=\zeta_k (\x)$ is the concentration of the compounds
$\bar \x$ of size $k$, the concentration of particles involved in these compounds equals \eqref{eqasscompbasmaes2},
that is
\begin{equation}
\label{eqasscompbasmaes2rep}
c(x)=\int_{SX^{k-1}} \zeta_k (x, x_2, \cdots , x_k) \mu_k(x, dx_2 \cdots dx_k).
\end{equation}
\end{proposition}
\begin{proof}
Let firstly $k=2$. The arbitrary measure $M_2$ on $SX^2$ can be given by the pair of measures $M^d$ and $M^{nd}$
(the subscripts $d$ and $nd$ stand for diagonal and non-diagonal parts), where $M^d$ is a measure on
the diagonal $D=\{(x,x): x\in X\}$ and $M^{nd}$ is a symmetric measure on $X^2\setminus D$, so that,
for a symmetric function $f$,
\begin{equation}
\label{eq1kernstoich}
\int_{SX^2} f(x,y)M_2(dx dy) =\frac12 \int_{X^2} f(x,y) M^{nd}(dx dy) +\int f(x,x) M^d(dx).
\end{equation}
Assuming that $M_2$ has absolutely continuous (with respect to $M$) projections on $X$ means that there exist
a kernel $\mu^{nd}(x,dy)$ with $\mu^{nd}(x, \{x\})=0$ and a function $\om (x)$ such that
\[
M^{nd} (dx dy) =M(dx) \mu^{nd}(x, dy), \quad M^d(dx) =\om (x) M(dx).
\]
Then clearly \eqref{eq1kernstoich} becomes
\begin{equation}
\label{eq2kernstoich}
\int_{SX^2} f(x,y)M_2(dx dy) =\int_X \frac12 \left[ \int f(x,y) \mu_2(x,dy)\right] M^d(dx)
\end{equation}
with
\[
\mu_2(x,dy)= \mu^{nd} (x,dy)+2\om (x) \de (x-y).
\]
Moreover, the amount of particles in a neighborhood $dx$ of a point $x$ entering the compounds is
\begin{equation*}
\begin{split}
\int_{dx} \int_X& \zeta(x,y) M^{nd} (dx dy) +2\int_{dx} \zeta (x,x) M^d(dx)\\
 &=\int_{dx} M(dx) \left[\int_X \zeta (x,y) \mu^{nd} (x,dy)+2\zeta(x,x)\om (x)\right]
\end{split}
\end{equation*}
(a particle at $x$ is used twice in the compound $\zeta(x,x)$, hence the coefficient $2$ at the second term).
Hence the concentration, which is the density with respect to $M(dx)$ is
\[
c(x)=\int_X\zeta (x,y) \mu^{nd} (x,dy)+2\zeta(x,x)\om (x)=\int \zeta (x,y)\mu_2(x, dy),
\]
as required.

Now let $k=3$. Then an arbitrary measure $M_3$ on $SX^3$ can be given by the triple $M^d$, $M^{nd}$ and $M^{int}$,
where $M^d$ is a measure on
the diagonal $D^3=\{(x,x,x): x\in X\}$,
$M^{nd}$ is a symmetric measure on $X^3\setminus D^{23}$, where
\[
D^{23}=\{(x_1,x_2,x_3): \exists i,j: x_i=x_j \},
\]
and $M^{int}$ is a measure on $X^2$ (not necessarily symmetric, that counts the triples $(x,x,y)$ with $y\neq x$)
 so that for a symmetric function $f$,
\begin{equation}
\begin{split}
 \int_{SX^3} f(x_1,x_2, x_3) M_3(dx_1 dx_2 dx_3)
 =&\frac16 \int_{X^3} f(x_1, x_2, x_3) M^{nd}(dx_1 dx_2 dx_3)\\
\label{eq3kernstoich}  +&\int_X f(x,x,x) M^d(dx)+\int_{(X\times X)\setminus D} f(x,x,y)
M^{int}(dx dy).
\end{split}
\end{equation}
Assuming that $M_2$ has absolutely continuous (with respect to $M$) projections on $X$ implies that
all three measures above have this property and the proof of the statement can be performed separately for each of them.
For $M^{nd}$ and $M^d$ it is literally the same as for the case $k=2$. Let us consider a more subtle case
of the measure $M^{int}$. Denoting by $\mu_{12}$ and $\mu_{21}$ the kernels arising from the projections
of $M^{int}$ on the first and the second coordinate (note that they are not symmetric, as the first coordinate
describes the pairs of identical particles), we have
\[
M^{int}(dx dy)= M(dx) \mu_{21}(x,dy)=  M(dy)\mu_{12}(y,dx)
\]
and therefore also
\[
M^{int}(dx dy)= \frac23M(dx) \mu_{21}(x,dy)+\frac13 M(dy)\mu_{12}(y,dx).
\]
Consequently, defining the kernel
\begin{equation}
\label{eq4kernstoich}
\mu(x,dy\, dz)=2\mu_{21}(x,dy)\de (z-x) +\mu_{12}(x,dy) \de (z-y),
\end{equation}
allows one to write
\[
\int_{X\times X\setminus D} f(x,x,y) M^{int}(dx dy)
=\frac13 \int_X \left[\int_{SX^2}f(x,y,z) \, \mu(x, dy\, dz)\right]  M(dx).
\]
Moreover, the amount of particles in a neighborhood $dx$ of a point $x$ entering the compounds
containing precisely two identical particles equals
\[
\int_{dx} \int_X \zeta(x,y,y) M^{int} (dy dx) +2\int_{dx} \int_X \zeta (x,x,y) M^{int}(dx dy)
\]
\[
=\int_{dx} M(dx) \left[\int_X \zeta (x,y,y) \mu_{12} (x,dy)+2\int_X \zeta(x,x,y)\mu_{21} (x,dy)\right].
\]
Hence the concentration, which is the density with respect to $M(dx)$ is
\[
c(x)=\int_{SX^2} \zeta (x,y,z)\mu (x, dy \, dz),
\]
as required.

Larger $k$ are analyzed similarly, but requires understanding of the structure of measures on $SX^k$
discussed below.
\end{proof}

Recall that a partition of a natural number $k$ is defined as its representation
as a sum of non-vanishing terms (with the order of terms irrelevant), i.e. as
\begin{equation}
\label{eqdefparti}
k=N_1+2N_2 +\cdots +jN_j
\end{equation}
with a $j>0$, where $N_l$ is the number of terms in the sum that equal $l$. Graphically these partitions are described by the
so-called Young schemes. For a partition (or a Young scheme) \eqref{eqdefparti} let us defined the extended diagonal
$D^{N_1, \cdots ,N_j}$ as a subset of the product $X^{N_1+ \cdots +N_j}$ such that at least two of the coordinates
$(x_1, x_2, \cdots , x_{N_1+\cdots + N_j})$ coincide. The following fact is then more or less straightforward.

\begin{proposition}
\label{propstructmessymm}
An arbitrary Borel measure $M_S$ on $SX^k$ can be uniquely specified by a collection of measures
$M^{N_1, \cdots ,N_j}$ on $X^{N_1+ \cdots +N_j}\setminus D^{N_1, \cdots ,N_j}$ which are symmetric for permutations
inside the group of arguments in each $X^{N_l}$ and which are parametrized by all partitions
 \eqref{eqdefparti}, so that for a symmetric function $f$ on $X^k$
\begin{equation}
\begin{split}
\int_{SX^k} f(x) M_S (dx) =&\sum_{N_1, \cdots, N_j}\frac{1}{N_1! \cdots N_j!}
 \int_{X^{N_1+ \cdots +N_j}\setminus D^{N_1, \cdots ,N_j}} M^{N_1, \cdots ,N_j}(dx_1 \cdots dx_{N_1+\cdots
 +N_j})\\
\label{eqdefmesonsym1} &\times f(x_1, \cdots, x_{N_1},
 \cdots , \underbrace{x_{N_1+\cdots +N_{l-1}+m}, \cdots , x_{N_1+\cdots +N_{l-1}+m}}_{l {\text times}}, \cdots)
 \end{split}
\end{equation}
(the arguments coincide in each group entering the partition),
the sum being over all partitions
 \eqref{eqdefparti} of $k$.
If, additionally, the projection of $M_S$ on $X$ is absolutely continuous with respect to a measure $M(dx)$, that is each measure
$M^{N_1, \cdots ,N_j}$ is absolutely continuous with respect to each arguments, then it can be presented in
$N_1+\cdots +N_j$
equivalent forms:
\[
M^{N_1, \cdots ,N_j}(dx_1 \cdots  dx_{N_1+\cdots +N_j})=M(dx_{N_1+\cdots +N_{l-1}+m})
\]
\begin{equation}
\label{eqdefmesonsym2}
 \mu_l^{N_1, \cdots ,N_j} (x_{N_1+\cdots +N_{l-1}+m}, dx_1  \cdots  d\check{x}_{N_1+\cdots +N_{l-1}+m}
 \cdots  dx_{N_1+\cdots +N_j}),
\end{equation}
where $\check{x}_p$ denotes, as usual, the absence of $x_p$ in the sequence of arguments,
$\mu_l$ are some stochastic kernels and $m \in \{1, \cdots N_l\}$, or more symmetrically as
\[
M^{N_1, \cdots ,N_j}(dx_1 \cdots  dx_{N_1+\cdots +N_j})=\sum_{l=1}^j \frac{l}{k}\sum_{m=1}^{N_l}
 M(dx_{N_1+\cdots +N_{l-1}+m})
\]
\begin{equation}
\label{eqdefmesonsym3}
\mu_l^{N_1, \cdots ,N_j} (x_{N_1+\cdots +N_{l-1}+m}, dx_1  \cdots  d\check{x}_{N_1+\cdots +N_{l-1}+m} \cdots  dx_{N_1+\cdots +N_j}).
\end{equation}
\end{proposition}

The numerators $l$ in \eqref{eqdefmesonsym3} reflect the number of identical particles entering a compound, thus presenting the analogs
of stoichiometric coefficients.

\end{document}